\documentclass[11pt]{amsart}

\usepackage{latexsym}
\usepackage{amssymb}
\usepackage{amsmath}
\usepackage{color}

\newtheorem{theorem}{Theorem}[section]
\newtheorem{lemma}[theorem]{Lemma}
\newtheorem{proposition}[theorem]{Proposition}
\newtheorem{corollary}[theorem]{Corollary}
\newtheorem{definition}[theorem]{Definition}

\newtheorem{remark}[theorem]{Remark}

\newtheorem{question}[theorem]{Question}

\newcommand\nph{\varphi}

\newcommand{\nrm}[1]{\| #1\|}

\newcommand{\cl}[1]{\mathcal{#1}}
\newcommand{\bb}[1]{\mathbb{#1}}

\begin{document}

\title{Measurable no-signalling correlations}

\author[G. Baziotis]{Georgios Baziotis}
\address{
School of Mathematical Sciences, University of Delaware, 501 Ewing Hall,
Newark, DE 19716, USA}
\email{baziotis@udel.edu}

\author[I. G. Todorov]{Ivan G. Todorov}
\address{
School of Mathematical Sciences, University of Delaware, 501 Ewing Hall,
Newark, DE 19716, USA}
\email{todorov@udel.edu}

\author[L. Turowska]{Lyudmila Turowska}
\address{Department of Mathematical Sciences, Chalmers University
of Technology and the University of Gothenburg, Gothenburg SE-412 96, Sweden}
\email{turowska@chalmers.se}

\date{25 September 2024}

\begin{abstract} 
We study no-signalling correlations, defined over 
a quadruple of second countable compact Hausdorff spaces. 
Using operator-valued information channels over 
abstract alphabets, we define the subclasses of 
local, quantum spatial and quantum commuting measurable 
no-signal\-ling correlations. En route, we establish 
measurable versions of the Stinespring's Dilation Theorem.
We define values of measurable non-local games of local, quantum 
spatial and 
quantum commuting type, as well as inner versions thereof, 
and show how the asymptotic values of a finite non-local game
can be viewed as special cases of the corresponding 
inner values of a measurable game, canonically associated 
with the given finite game. 
\end{abstract}

\maketitle

\tableofcontents


\section{Introduction}\label{s_intro}

Non-local games are played cooperatively by 
two players against a Verifier, aiming to convince the latter 
of the joint possession of a certain predescribed knowledge. 
In each round of the game, the Verifier selects a pair of questions 
from given finite sets $X$ and $Y$ according to a probability 
distribution on $X\times Y$, and sends one to each of 
the players. The players respond, without communicating with each other, 
with answers selected from finite sets $A$ and $B$, 
and upon their receipt by the Verifier, the latter checks if the quadruple 
$(x,y,a,b)$ formed by questions and answers satisfy 
a given rule predicate
$\lambda : X\times Y \times A\times B\to \{0,1\}$ -- 
if $\lambda(x,y,a,b) = 1$ 
the players win, and if not -- they 
lose, the round. 

Non-local games have been extensively studied from the perspectives of 
complexity theory, quantum information theory and operator algebras, 
with diverse applications varying from theoretical physics to pure 
mathematics; see 
\cite{csuu, cameron, dinur, mr_homom, lmr, lmprsstw, jnvwy, psstw, trevisan}
and the references therein.
In the event where the players do not possess a perfect strategy 
(that is, a strategy whose use yields a win in any round), 
the optimal winning probability, given a strategy type 
available to the players,  and its asymptotic version, 
become essential. 
The latter parameters are known as the 
game's value and the game's asymptotic value, respectively. 

The asymptotic value is computed by taking a limit 
of the normalised values 
along the parallel (IID) round repetition. 
It becomes natural to consider the infinite parallel
repetition of a finite non-local game
as a single infinite game, and to attempt defining 
a global pool of strategies which includes the 
strategies of an $n$-fold repetition as a subclass. 

In the present paper, we propose a setup that permits this. 
In fact, the setup we suggest goes beyond the IID model altogether, 
hosting games with memory; their
rule predicates (resp. probability distributions 
used by the Verifier to select questions)
in an $n$-round scenario
are not given by the $n$-fold product of the rule predicate
(resp. the $n$-fold product of the probability distribution) of the first round. 
Instead, the games whose study we initiate 
can be viewed as Borel subsets of $X\times Y\times A\times B$, where 
$X$, $Y$, $A$ and $B$ are 
second countable compact Hausdorff spaces (most typically,
Cantor spaces), along with a designated Borel 
probability measure on the direct product topological space $X\times Y$. 

In the usual non-local game case 
(where the sets $X$, $Y$, $A$ and $B$ are finite), 
the strategies of the players are bipartite no-signalling correlations with 
conditional probabilities $p(a,b|x,y)$ that represent the 
likelihood that the players 
respond with the pair $(a,b)$ of answers, when they are provided with 
the pair $(x,y)$ of questions. 
The main correlation types of interest are the 
local, quantum spatial 
and quantum commuting ones \cite{lmprsstw}, leading to 
corresponding (asymptotic) local, quantum and quantum commuting, values. 
By viewing a no-signalling correlation as an information channel 
from $X\times Y$ to $A\times B$, we define the class of measurable 
no-signalling correlations,
lifting simultaneously to the measurable setup 
the three aforementioned correlation types.

The latter is achieved by defining and proving some required 
properties of 
operator-valued information channels over 
abstract alphabets. 
We show that, under certain conditions, 
an operator-valued information channel can be dilated to 
a projection-valued information channel. 
In fact, this can be viewed as a consequence of a 
measurable version of the fundamental Stinespring Theorem 
(see \cite{arveson-acta} and \cite{Pa}) that we establish herein. 
We further show that operator-valued channels with commuting ranges 
can be multiplied to obtain a new operator-valued channel on the 
product of the two Borel $\sigma$-algebras. The dilation results 
in this case can be extended to show that the product operator-valued 
channel admits measurable dilations that are themselves products of
measurable families of spectral measures. 
In the simplest of cases, our approach yields a 
short operator-theoretic proof of the existence of a 
product spectral measure (resp. quantum probability measure)
arising from two spectral measures (resp. quantum probability measures)
with commuting ranges, under a metrisability assumption on the
underlying spaces; we note that a more general result 
was obtained by Birman and Solomyak in \cite{bs1} (see also \cite{bs2}). 

The paper is organised as follows. 
In Section \ref{s_ovic}, we 
extend the scalar-valued information theoretic 
setup of \cite{kakihara} to the operator-valued case,
establishing a correspondence between the $\cl B(H)$-valued information channels 
from a compact Hausdorff space $X$
into a compact Hausdorff space $A$ and the Borel families of 
unital completely positive maps defined in the space $C(A)$ 
of all complex-valued continuous functions on $A$.

In Section \ref{s_prod}, we establish the existence of product
operator-valued information channels, while Section \ref{s_dilres} is devoted to 
dilation results. 
The starting point is a measurable family $(\phi_x)_{x\in X}$
of unital completely positive maps defined on a compact Hausdorff space $X$. 
We prove two dilation results, in the case $X$ is equipped with 
a Borel measure $\nu$:
when $\nu$ is a 
continuous measure, we explicitly construct a $\nu$-measurable 
family $(\pi_x)_{x\in X}$ of *-homomorphisms 
which is a simultaneous dilation of the family 
$(\phi_x)_{x\in X}$ for every $x\in X$
while, if $\nu$ is arbitrary, 
the family $(\pi_x)_{x\in X}$ dilates $(\phi_x)_{x\in X}$ 
simultaneously almost everywhere. 
We note that the simultaneous dilations are subsequently used to prove disambiguation for values of measurable games: in the finite case, no-signalling correlations of (quantum and) quantum commuting 
type can be equivalently defined using either families of commuting POVM's 
or families of commuting PVM's. The fact that these two ways are indeed equivalent
does not readily follow in the (infinite) measurable case. 

In Section \ref{s_mnsc}, we define measurable no-signalling correlations, as 
well as their subclasses of local, quantum spatial and quantum commuting type. These types are sufficient for the definition of the main game values; 
we therefore leave the investigation of the 
approximately quantum measurable correlation type for future work, 
as it also requires appropriate notions of convergence, which is developed 
elsewhere (see \cite{baziotis}). We show that each of these types is 
preserved under the formation of products. 

In Section \ref{s_val}, we define values of measurable games of 
local, quantum spatial and quantum commuting type. 
Given topological homeomorphisms of the base spaces $X$, $Y$, $A$ and $B$, 
we define the inner value of a measurable game, which specialises 
to the asymptotic value of a finite game, when the latter is 
viewed as a measurable game. 
We point out that in the typical case, measurable games are defined over
Cantor spaces, and in the latter situation, the most useful choice
of reference homeomorhisms is that of the shifts. 
We give an example of a game with a genuine (one-site) memory, 
whose IID asymptotic value differs from its inner value. 
Finally, in Section \ref{s_questions} we collect several questions
arising from our work. 

\smallskip

We finish this section with setting basic notation and 
introducing some notions that will be used subsequently.
For a subset $\cl E$ of a vector space $\cl U$, we 
let $[\cl E]$ be the linear span of $\cl E$. 
For a Hilbert space $H$, we write $\cl B(H)$ for the C*-algebra of 
all bounded linear operators on $H$, and 
$\cl T(H)$ for the trace class ideal; we recall the canonical identification 
$\cl T(H)^* = \cl B(H)$. 
All topological spaces will be assumed to be second countable. 
We denote by $C(X)$ the unital 
C*-algebra of all continuous complex-valued functions on a compact Hausdorff space $X$,
and by $\frak{A}_X$ the Borel $\sigma$-algebra of $X$. 
We write $M(X)$ for the operator space of all complex Borel measures on $X$, and 
$P(X)$ for the (convex) set of all probability measures on $X$. 
Given $x\in X$, we let $\delta_x\in M(X)$ be 
the point mass at $x\in X$. 
Note the completely isometric identification $C(X)^* = M(X)$. 

If $\cl A$ is a C*-algebra, we write $\cl A^+$ for the cone of all positive elements of $\cl A$. 
A linear map $\phi : \cl A\to \cl B$ between unital C*-algebras is called 
\emph{positive}, if $\phi(\cl A^+)\subseteq \cl B^+$, and \emph{completely positive}, if the map 
$\phi^{(n)} : M_n(\cl A)\to M_n(\cl B)$ between the C*-algebras of $n$ by $n$ matrices over 
$\cl A$ and $\cl B$, respectively, given by $\phi^{(n)}((a_{i,j})_{i,j}) = (\phi(a_{i,j}))_{i,j}$, 
is positive for every $n\in \bb{N}$. 
The map $\phi$ is called \emph{unital} if $\phi(1) = 1$.


\section{Operator-valued information channels}\label{s_ovic}

In this section, we define an operator-valued version of information 
channels over abstract alphabets, and provide a characterisation thereof. 
Let $A$ be a compact Hausdorff space and $H$ be a separable Hilbert space. 
Recall that a \emph{quantum probability measure (QPM)} on $A$ with values in $\cl B(H)$ 
\cite{ffp, Pa} is a map
$E : \frak{A}_A\to \cl B(H)^+$ such that 
\begin{itemize}
\item[(i)] $E(A) = I_H$;
\item[(ii)] 
$E\left(\cup_{i=1}^{\infty}\alpha_i\right) = \sum_{i=1}^{\infty} E(\alpha_i)$ 
in the weak* topology, whenever
$(\alpha_i)_{i\in \bb{N}}\subseteq \frak{A}_A$
and $\alpha_i\cap \alpha_j = \emptyset$, $i\neq j$.
\end{itemize}
A quantum probability measure $E$ on $A$ with values in $\cl B(H)$ is called a 
\emph{spectral measure} (or a \emph{projection-valued quantum probability measure}) if 
$E(\alpha)$ is a projection for every $\alpha\in \frak{A}_A$ 
(equivalently, if $E(\alpha\cap \beta) = E(\alpha)E(\beta)$ for all $\alpha,\beta \in \frak{A}_A$). 
We write $Q(A;H)$ (resp. $P(A;H)$) for the set of all 
QPM's (resp. spectral measures) on $A$ with values in $\cl B(H)$. 
In the case where $A$ is finite, the elements of $Q(A;H)$ (resp. $P(A;H)$) 
are called \emph{positive operator-valued measures (POVM's)}
(resp. \emph{projection-valued measures (PVM's)}). Given $E\in Q(A;H)$ and vectors $\xi,\eta\in H$, we denote by $E_{\xi,\eta}$ the complex measure on $X$, given by 

$$E_{\xi,\eta}(\alpha)=\langle E(\alpha)\xi,\eta\rangle, \hspace{0.3cm} \alpha\in\frak{A}_A. $$

Let $X$ be a(nother) compact Hausdorff space. 
We will write $\frak{A}_{XA}$ for the Borel $\sigma$-algebra, generated by the sets of the form
$\chi\times\alpha$, $\chi\in \frak{A}_X$, $\alpha\in \frak{A}_A$.  
Given a subset $M\in \frak{A}_{XA}$ and $x\in X$, we let 
$M_x = \{a\in A : (x,a)\in M\}$ be the $A$-section of $M$ along $x$ (note that $M_x\in \frak{A}_A$ for every $x\in X$). 
A function $F : X\to \cl B(H)$ is called 
\emph{weakly measurable} if the function
\begin{equation}\label{eq_Fxieta}
x\mapsto \langle F(x)\xi,\eta\rangle
\end{equation}
is Borel for all vectors $\xi,\eta\in H$. 
Given a Borel measure $\mu$ on $X$, 
let $\frak{A}_{\mu}$ be the completion of $\frak{A}_X$ with respect to $\mu$.
The function $F : X\to \cl B(H)$ is called 
\emph{weakly $\mu$-measurable} if the function
(\ref{eq_Fxieta}) 
is \emph{$\mu$-measurable}, that is, measurable with respect to 
the $\sigma$-algebra $\frak{A}_{\mu}$, 
for all vectors $\xi,\eta\in H$.
By polarisation, $F$ is weakly measurable (resp. 
weakly $\mu$-measurable) if and only if 
the function $x\mapsto \langle F(x)\xi,\xi\rangle$ is measurable 
(resp. $\mu$-measurable) for all vectors $\xi\in H$.

We recall that an information channel over the pair $(X,A)$ of abstract alphabets is a
family $(p_x)_{x\in X}$ of probability measures on $A$, such that, for every 
$\alpha\in \frak{A}_A$, the function $x\mapsto p_x(\alpha)$ is Borel \cite{kakihara}. 
We will need an extension of this notion, in the presence of a reference 
measure on the input alphabet $X$. Namely, 
given a Borel probability measure $\mu$ on $X$, 
a \emph{$\mu$-information channel} over the pair $(X,A)$ 
is a family family $(p_x)_{x\in X}$ of probability measures on $A$, 
such that, for every 
$\alpha\in \frak{A}_A$, the function $x\mapsto p_x(\alpha)$ is
$\mu$-measurable.
We note that every information channel is automatically a $\mu$-information 
channel.

A useful operator-valued extension of this concept can be defined as follows.

\begin{definition}\label{d_opval}
Let $X$ and $A$ be compact Hausdorff spaces, $H$ be a Hilbert space 
and $E = (E(\cdot|x))_{x\in X}$ be a family 
of quantum probability measures on $A$ with values in $\cl B(H)$.
We call $E$ 
an \emph{operator-valued information channel} from $X$ to $A$
with values in $\cl B(H)$ if, for every $\alpha\in \frak{A}_A$, the function 
\begin{equation}\label{eq_measalpE}
X\to \cl B(H); \ \ x \mapsto E(\alpha|x)
\end{equation}
is weakly measurable. 

Given a Borel probability measure $\mu$ on $X$, 
we call $E$ an \emph{operator-valued $\mu$-information channel} from $X$ to $A$
with values in $\cl B(H)$
if, for every $\alpha\in \frak{A}_A$, the function (\ref{eq_measalpE}) is weakly $\mu$-measurable.
\end{definition}

A $\cl B(H)$-valued information channel $E$ from $X$ to $A$ will be written $E : X\to (A;H)$ for short; if $H = \bb{C}$, we write $E : X\to A$.
Sometimes it will be convenient to set $E_x(\alpha) = E(\alpha|x)$, $x\in X$, $\alpha\in \frak{A}_A$. 
Given $E:X\mapsto(A;H)$ and a unit vector $\xi\in H$ we denote by $(E_{x,\xi,\xi})_{x\in X}$ the scalar valued information channel, given by 
$$E_{x,\xi,\xi}(\alpha)=\langle E(\alpha|x)\xi,\xi\rangle, \hspace{0.2cm} \alpha\in\frak{A}_A, x\in X. $$
\par We will shortly characterise operator-valued channels in terms of 
canonical completely positive maps. We will need the 
following rather well-known result (see e.g. \cite[Proposition 4.5]{Pa}). Since we could not 
find a reference for its proof, we include a sketch for the convenience of the reader.

\begin{theorem}\label{qpm}
Let $A$ be a compact Hausdorff space, $H$ be a Hilbert space and 
$E:\mathfrak{A}_A\rightarrow \mathcal{B}(H)$ be a quantum probability measure 
(resp. spectral measure). 
Then there exists a unique unital completely positive map (resp. *-homomorphism) 
$\phi_E : C(A) \rightarrow \mathcal{B}(H)$, such that 			
$$\langle \phi_E(f)\xi,\eta\rangle = \int_A f dE_{\xi,\eta}, \hspace{0.6cm} f\in C(A),\ \xi,\eta\in H.$$	
Conversely, if $\phi : C(A) \rightarrow\mathcal{B}(H)$
is a unital completely positive map (resp. *-homomorphism) then 
there exists a unique quantum probability measure
(resp. spectral measure) $E:\mathfrak{A}_A \rightarrow\mathcal{B}(H)$ such that $\phi = \phi_E$.  
\end{theorem}

\begin{proof}
Assume that $E$ is a quantum probability measure. 
For $\xi,\eta\in H$, let 
$\phi_{\xi,\eta}\in C(X)^*$ be the functional, 
given by $\phi_{\xi,\eta}(f)=\int_X f dE_{\xi,\eta}$. 
For a fixed $f\in C(X)$, let $\mu_f:H\times H\rightarrow\mathbb{C}$ be the map, 
defined by letting $\mu_f(\xi,\eta)=\phi_{\xi,\eta}(f)$.
A straightforward verification shows that 
$\mu_f$ is a bounded sesquilinear form and
		\begin{align*}
			\begin{split}
				\lvert\mu_f(\xi,\eta)\rvert & =\Big\lvert\int_X fdE_{\xi,\eta}\Big\rvert\leq\nrm{f}_\infty\nrm{E_{\xi,\eta}}\leq\nrm{f}_\infty\|\xi\|\|\eta\|.
			\end{split}
		\end{align*}
Let $\phi_E(f)\in\mathcal{B}(H)$ be such that 
$\mu_f(\xi,\eta) = \langle \phi_E(f)\xi,\eta\rangle$, $\xi,\eta\in H$.
Since, for any $\xi\in H$, the measure $E_{\xi,\xi}$ is positive, the map $f\mapsto \phi_E(f)$ is positive.
Since the domain of $\phi_E$ is commutative, $\phi_E$ is completely positive
\cite[Theorem 3.11]{Pa}.  

Conversely, suppose that $\phi:C(X)\rightarrow\mathcal{B}(H)$ is a unital completely positive map. 
By Stinespring's Dilation Theorem, there exists a Hilbert space $K$, a unital *-homomorphism $\pi:X\rightarrow \mathcal{B}(K)$ and an isometry  $V:H\rightarrow K$,
such that $\phi(f)=V^*\pi(f)V$, $f\in C(A)$. 
Since $\pi:C(X)\rightarrow\mathcal{B}(H)$ is a $*$-homomorphism, 
by the Spectral Theorem for commutative C$^*$-algebras, 
there exists a unique spectral measure $P:\mathfrak{A}_X\rightarrow\mathcal{B}(K)$, such that 
		\begin{align*}
			\langle\pi(f)\xi,\eta\rangle = \int_X fdP_{\xi,\eta}, \hspace{0.6cm} \xi,\eta\in K, f\in C(X).
		\end{align*}
Define $E(\delta) = V^*P(\delta) V$, $\delta\in \frak{A}_X$. 
The countable additivity of $E$ follows from that of $P$ and the fact that $V$ is an isometry.
For $\xi,\eta\in H$, we have 
$$				E_{\xi,\eta}(\delta)
    =\langle V^*P(\delta)V\xi,\eta\rangle = P_{V\xi,V\eta}(\delta), \ \ \ \delta\in\mathfrak{A}_X.$$
Further, 
$$\int_XfdE_{\xi,\eta} = \int_XfdP_{V\xi,V\eta} = \langle\pi(f)V\xi,V\eta\rangle =
\langle\phi(f)\xi,\eta\rangle.$$

It is a standard fact that if $E$ is a spectral measure 
then the assignment $\phi_E(f)=\int_X fdE$ defines a $*$-representation of $C(A)$. The converse part follows from the Spectral Theorem. 
\end{proof}

The following lemma will be needed for the consideration of product channels in Section \ref{s_prod}. 

\begin{lemma}\label{cqpm}	
Let $X$ be a compact Hausdorff space, $H$ a Hilbert space and $E$ be a quantum probability measure on $H$. 
Then $S \in ({\rm ran}\phi_E)'$ if and only if 
$SE(\delta) = E(\delta)S$ for every $\delta\in\mathfrak{A}_X$.	
\end{lemma} 

\begin{proof}
Suppose that $S \in ({\rm ran}\phi_E)'$. 
Let $f\in C(X)$ and $\xi,\eta\in H$; then 
$$\langle S\phi_E(f)\xi,\eta\rangle = \langle\phi_E(f)\xi,S^*\eta\rangle=\int_X f dE_{\xi,S^*\eta},$$ 
while
$$\langle \phi_E(f)S\xi,\eta\rangle=\int_X f dE_{S\xi,\eta}.$$
By the uniqueness clause of the Riesz Representation Theorem, 
for any $\delta\in\mathfrak{A}_X$ we have that 
$$\langle E(\delta)S\xi,\eta\rangle  
= E_{S\xi,\eta}(\delta) = 
E_{\xi,S^*\eta}(\delta) = \langle E(\delta)\xi,S^*\eta\rangle = 
\langle SE(\delta)\xi,\eta\rangle;$$
thus, $SE(\delta) = E(\delta)S$.  

The converse direction follows by reversing the steps in the previous paragraph.
\end{proof}
	
For a Hilbert space $H$, we will denote by $\frak{B}(X;H)$ the space of all
bounded weakly measurable functions $F : X\to \cl B(H)$, where the space $X$ is equipped with the Borel $\sigma$-algebra $\frak{A}_X$.
Similarly, we write $\frak{B}_\mu(X;H)$ for the space of all
bounded weakly $\mu$-measurable functions $F : X\to \cl B(H)$.
We note that, if $H$ is separable, $\frak{B}(X;H)$ and 
$\frak{B}_{\mu}(X;H)$ are unital C*-algebras with respect to pointwise addition and multiplication, the norm given by
$\|F\| := \sup_{x\in X}\|F(x)\|$, 
and involution given by $F^*(x) := F(x)^*$, $x\in X$.
The following theorem is a generalisation of Theorem \ref{qpm} 
to operator-valued information channels.

\begin{theorem}\label{th_Borel}
Let $X$ and $A$ be compact Hausdorff spaces, 
$H$ be a separable Hilbert space, and 
$E:X\to(A;H)$ be an operator-valued (resp. projection-valued) information channel. 
Then the map $\Phi_E : C(A)\rightarrow\frak{B}(X;H)$, given by 
$$\Phi_E(f)(x) = \phi_{E_x}(f), \hspace{0.3cm} f\in C(A), x\in X,$$
is well-defined, unital and completely positive (resp. a well-defined unital *-homomor\-phism).

Conversely, if $\Phi : C(A)\rightarrow \frak{B}(X;H)$ is a unital completely positive map (resp. *-homomorphism) 
then there exists an operator-valued (resp. projection-valued) information channel 
$E:X\to(A;H)$ such that $\Phi = \Phi_E$.
\end{theorem}

\begin{proof}
Let $E:X\to(A;H)$ be an operator-valued information channel and, using the notation from 
Theorem \ref{qpm}, set $\phi_x = \phi_{E_x}$ for brevity.
Fix $f\in C(A)$. 
To see the measurability of the function $x\mapsto\phi_x(f)$, 
approximate $f\in C(A)$ pointwise by a sequence $(f_n)_{n\in\bb{N}}$ 
of uniformly bounded Borel functions of the form 
$f_n = \sum_{i=1}^{m_n} \lambda_{n,i}\chi_{\delta_{n,i}}$. 
For each $\xi\in H$ and $x\in X$, 
we have 
$$\int_A f_n dE_{x,\xi,\xi} = 
\sum_{i=1}^{m_n}\lambda_{n,i} E_{x,\xi,\xi}(\delta_{n,i})
= \sum_{i=1}^{m_n} \lambda_{n,i} \langle E(\delta_{n,i}|x)\xi,\xi\rangle, \ \ \ n\in \bb{N}.$$ 
Thus, the function 
$x\mapsto \int_A f_n dE_{x,\xi,\xi}$ is Borel for every $n\in \bb{N}$. On the other hand, 
for a fixed $x\in X$, by 
Lebesgue Dominated Convergence Theorem, 
$$\int_A f dE_{x,\xi,\xi}=\lim_{n\to \infty} \int_Af_ndE_{x,\xi,\xi};$$
it follows that 
the function $x\mapsto \int_AfdE_{x,\xi,\xi}$ is Borel. 
The measurability of the function $x\mapsto \int_AfdE_{x,\xi,\eta}$ follows by polarisation.

By Theorem \ref{qpm}, $\phi_x$ is unital and (completely) positive for every $x\in X$. 
It follows that $\Phi_E$ is a well-defined positive map; since its domain is commutative, 
it is completely positive \cite[Theorem 3.11]{Pa}.
	
Conversely, let $\Phi : C(A)\rightarrow\frak{B}(X;H)$ be a unital completely positive map. 
Let $\phi_x : C(A)\to \cl B(H)$ be the map, given by $\phi_x(f) = \Phi(f)(x)$, $x\in X$. 
The complete positivity of $\Phi$ implies the complete positivity of $\phi_x$ for every $x\in X$. 
By Theorem \ref{qpm}, for any $x\in X$ there exists a QPM $E_x = E(\cdot|x)$ satisfying 
$\phi_x = \phi_{E_x}$, $x\in X$. 
Fix $\xi\in H$ and $\alpha\in \frak{A}_A$. 
Let $(f_n)_{n\in \bb{N}}\subseteq C(A)$ be a uniformly bounded sequence, such that 
$f_n(a)\to \chi_{\alpha}(a)$ for every $a\in A$, \cite[Theorem 11.6, Remark 11.7]{kechris}.
Then 
$$\langle E(\alpha|x)\xi,\xi\rangle = 
\lim_{n\to \infty} \int_A f_n(a) dE_{x,\xi,\xi}(a) = 
\lim_{n\to \infty} \langle \phi_x(f_n)\xi,\xi\rangle.$$
As the functions $x\mapsto \phi_x(f_n)$ are 
weakly measurable, the function $x\mapsto  \langle E(\alpha|x)\xi,\xi\rangle$
is Borel. A polarization argument gives the Borel measurability of 
$x\mapsto \langle E(\alpha|x)\xi,\eta\rangle$  for $\xi,\eta\in H$. 
		
The arguments regarding projection-valued operator channels follow
a direct modification of the completely positive case, using the fact that 
the maps $\phi_x$ are unital *-homomorphisms if and only if the quantum probability measures $E_x$ are projection-valued. 
	\end{proof}
	
From now on we will always assume without mentioning that $A$ is second countable.
Given a unital C*-algebra $\cl A$, a Hilbert space $H$, and completely positive maps 
$\phi_x : \cl A\to \cl B(H)$, $x\in X$, we call the family $(\phi_x)_{x\in X}$ \emph{Borel} if 
the function $x\to \langle \phi_x(a)\xi,\eta\rangle$ is Borel measurable for all $a\in \cl A$ and all 
$\xi,\eta\in H$.

\begin{remark}\label{r_measfamm}
\rm 
With the notation in Theorem \ref{th_Borel}, the same 
statements remain true if we replace $\frak{B}(X;H)$ with 
$\frak{B}_\mu(X;H)$, and operator-valued information channels with 
operator-valued $\mu$-information channels. 
In particular, 
the following are equivalent:
\begin{itemize}
\item[(i)]
$E$ is an operator-valued information channel
(resp. an operator-valued $\mu$-information channel);

\item[(ii)] 
the function $x\mapsto \phi_{E_x}(f)$ is 
weakly measurable (resp. weakly $\mu$-measurable) for every $f\in C(A)$. 
\end{itemize}
\end{remark}


\section{Products of channels}\label{s_prod}

In this section, we prove the existence of a natural product 
of operator-valued channels with commuting ranges and, as a consequence, 
the existence of a tensor-type product of operator-valued channels.
Similarly to Section \ref{s_ovic}, we will need an initial result concerning products of
quantum probability measures. 
In the case of spectral measures, this is a rather well-known fact 
(see \cite{schaefer} and \cite{bs1, bs2}); our proof is operator-theoretic
and, we believe, new even in the spectral measure case.

Recalling the notation from Section \ref{s_ovic}, 
write $\frak{B}(X) = \frak{B}(X,\bb{C})$; thus, $\frak{B}(X)$
is the C*-algebra of all bounded Borel complex-valued functions on $X$. 
For $f\in \frak{B}(X)$, let $F_f : M(X)\to \bb{C}$ be the linear functional, 
given by $F_f(\mu) = \int_X fd\mu$. It is clear that 
$$|F_f(\mu)|\leq \|f\|_{\infty}\|\mu\|, \ \ \ f\in \frak{B}(X), \mu\in M(X);$$
thus, $F_f\in M(X)^*$. Let  
$\iota : \frak{B}(X)\to C(X)^{**}$ be the contractive map, given by 
$\iota(f)(\mu) = F_f(\mu)$, $\mu\in M(X)$. 
We have, in fact, that 
$$\|f\|_{\infty} = \sup_{x\in X} |f(x)| = |\iota(f)(\delta_x)|\leq \|\iota(f)\|,$$
and hence $\iota$ is isometric. 
We hence suppress the use of $\iota$ and consider $\frak{B}(X)$ as a unital C*-subalgebra 
of $C(X)^{**}$. 
In the next lemma, $\frak{B}(X)$ is equipped with the restriction of the weak*-topology 
of $C(X)^{**}$.

\begin{lemma}\label{l_ext}
Let $X$ be a compact Hausdorff space, $H$ be a Hilbert space and 
$E : \mathfrak{A}_X\rightarrow \mathcal{B}(H)$ be a quantum probability measure. 
The map $\phi_E$ from Theorem \ref{qpm} has a weak* continuous extension (denoted in the same way) 
$\phi_E : \frak{B}(X)\to \cl B(H)$, such that 
$\phi_E(\chi_{\delta}) = E(\delta)$, $\delta\in \frak{A}_X$. 
\end{lemma}

\begin{proof}
Let $\phi_E^{**}: C(X)^{**}\to \cl B(H)^{**}$ be the second dual of the 
map $\phi_E$ from Theorem \ref{qpm}, and let $\cl E: \cl B(H)^{**}\to \cl B(H)$ be the canonical projection map. 
We have that the composition 
$\cl E\circ\phi_E^{**} : C(X)^{**} \to \cl B(H)$ is a weak* continuous map, 
and we continue to write $\phi_E$ for its restriction to $\frak{B}(X)$. 
As the restriction of a weak* continuous map, $\phi_E$ is weak* continuous. 

Let $\delta\in\mathfrak{A}_X$ and let $(f_i)_{i\in I}\subseteq C(X)$ be a net
with weak* limit $\chi_{\delta}$. 
Let $\xi,\eta\in H$; then 
\begin{align*}
\begin{split}
\langle\phi_E(\chi_\delta)\xi,\eta\rangle 
& =\lim_{i\in I} \langle\phi_E(f_i)\xi,\eta\rangle
= \lim_{i\in I} \langle f_i, E_{\xi,\eta}\rangle\\
& = \langle \chi_\delta, E_{\xi,\eta} \rangle
= \int_X\chi_\delta dE_{\xi,\eta} = \langle E(\delta)\xi,\eta\rangle.
\end{split}
\end{align*}
Since $\xi$ and $\eta$ were chosen arbitrarily, $\phi_E(\chi_\delta)=E(\delta)$.
\end{proof}

\begin{theorem}\label{qpmprod} 
Let $X$ and $Y$ be compact Hausdorff spaces and $E : \mathfrak{A}_X\rightarrow \mathcal{B}(H)$ and $F:\mathfrak{A}_Y\rightarrow \mathcal{B}(H)$ be 
quantum probability measures with commuting ranges. Then there exists a unique quantum probability measure 
$E\cdot F : \mathfrak{A}_X\otimes\mathfrak{A}_Y\rightarrow\mathcal{B}(H)$ such that
$$(E\cdot F)(\delta\times\theta)=E(\delta)F(\theta), \ \ \ \delta\in \mathfrak{A}_X, \theta\in \mathfrak{A}_Y.$$
Moreover, if $E$ and $F$ are spectral measures, then $E\cdot F$ is a spectral measure. 
\end{theorem}

\begin{proof}
Let 
$\phi_E:C(X)\rightarrow\mathcal{B}(H)$ and $\phi_F:C(Y)\rightarrow\mathcal{B}(H)$
be the completely positive maps associated with $E$ and $F$, respectively, via
Theorem \ref{qpm}. 
Since $E$ and $F$ have commuting ranges, Lemma \ref{cqpm} implies that 
the range of $\phi_E$ commutes with $F(\beta)$ for every $\beta\in \frak{A}_Y$. 
Another application of Lemma \ref{cqpm} implies that $\phi_E$ and $\phi_F$
have commuting ranges.

Since the two maps have commuting ranges,
there exists a completely positive map 
$\phi_E\cdot\phi_ F : C(X)\otimes_{\max}C(Y)\rightarrow \cl{B}(H)$, satisfying 
		\begin{align}
			(\phi_E\cdot \phi_F) (f\otimes g) = \phi_E(f)\phi_F(g),
   \ \ \ f\in C(X), g\in C(Y).
		\end{align}
  By nuclearity, 
  $$C(X)\otimes_{\max}C(Y) = C(X)\otimes_{\min}C(Y) = C(X\times Y),$$ 
  up to canonical *-isomorphisms, where an elementary tensor 
$f\otimes g$ from $C(X)\otimes_{\max}C(Y)$ corresponds to the 
function $(x,y)\mapsto f(x)g(y)$ in $C(X\times Y)$. 
  We thus view $\phi_E\cdot \phi_F$ as a completely positive map, 
  defined on $C(X\times Y)$. 
  By Theorem \ref{qpm}, there exists a quantum probability measure 
  $E\cdot F:\mathfrak{A}_X\otimes\mathfrak{A}_Y\rightarrow \mathcal{B}(H),$ such that $\phi_{E\cdot F} = \phi_E\cdot\phi_F$.

Let $\delta\in\mathfrak{A}_X$ and $\theta\in\mathfrak{A}_Y$. 
Using Lemma \ref{l_ext}, we have 
\begin{eqnarray*}
(E\cdot F)(\delta\times\theta) 
= 
\phi_{E\cdot F}(\chi_{\delta\times\theta}) 
& = & 
(\phi_E\cdot\phi_F) (\chi_\delta\otimes\chi_\theta)\\
& = & 
\phi_E(\chi_\delta)\phi_F(\chi_\theta)
=
E(\delta)F(\theta).
\end{eqnarray*}

 In the case where $E$ and $F$ are spectral measures, by Theorem \ref{qpm} we have that $\phi_E$ and $\phi_F$ are unital *-homomorphisms with commuting ranges and 
hence the map $\phi_E\cdot \phi_F$ is a unital *-homomorphism. 
Using Theorem \ref{qpm}, we obtain a spectral measure $E\cdot F:\mathfrak{A}_X\otimes\mathfrak{A}_Y\rightarrow \mathcal{B}(H),$ such that $\phi_{E\cdot F}=\phi_E\cdot\phi_F$.  
\end{proof}

It is possible to derive the existence of a product spectral measure 
as a corollary of Theorem \ref{qpmprod}; we note that a more general result 
was obtained by Birman and Solomyak in \cite{bs1} (see also \cite{bs2}).

	\begin{corollary}\label{prod_qclass}
		Let $X$, $Y$ be compact Hausdorff spaces, $H$, $K$ be Hilbert spaces and $E:\mathfrak{A}_X\rightarrow\mathcal{B}(H)$, $F:\mathfrak{A}_Y\rightarrow\mathcal{B}(K)$ be quantum probability measures. Then there exists a quantum probability measure $E\otimes F:\mathfrak{A}_X\otimes\mathfrak{A}_Y\rightarrow\mathcal{B}(H\otimes K)$ satisfying,
		\begin{align*}
			(E\otimes F)(\delta\times\theta)=E(\delta)\otimes F(\theta), 
		\end{align*}
		for any measurable rectangle $\delta\times\theta$. Moreover, if $E$ and $F$ are spectral measures, 
  then $E\otimes F$ is a spectral measure.
	\end{corollary}

 \begin{proof}
Consider the maps $\tilde{E}:\mathfrak{A}_X\rightarrow\mathcal{B}(H\otimes K)$ and $\tilde{F}:\mathfrak{A}_Y\rightarrow\mathcal{B}(H\otimes K)$, given by $\tilde{E}(\delta)=E(\delta)\otimes I_K$ and $\tilde{F}(\theta)=I_H\otimes F(\theta)$, for $\delta\in\mathfrak{A}_X$ and $\theta\in\mathfrak{A}_Y,$ respectively. 
It is straightforward that 
$\tilde E$ and $\tilde F$ are quantum probability measures with commuting ranges.
By Theorem \ref{qpmprod}, there exists a quantum probability measure $\tilde{E}\cdot\tilde{F}:\mathfrak{A}_X\otimes\mathfrak{A}_Y\rightarrow \mathcal{B}(H\otimes K)$ satisfying $$\tilde{E}\cdot\tilde{F}(\delta\times\theta)=\tilde{E}(\delta)\tilde{F}(\theta)=(E(\delta)\otimes I_K)(I_H\otimes F(\theta))=E(\delta)\otimes F(\theta). $$
We set $E\otimes F := \tilde{E}\cdot\tilde{F}$.
If $E$ and $F$ are spectral measures, it is straightforward that so are  $\tilde{E}$, and $\tilde{F}$, and hence so is $E\otimes F$.
\end{proof}

\begin{theorem}\label{cqc_prod}
Let $X,Y,A,B$ be compact Hausdorff  spaces, 
$H$ be a separable Hilbert space, 
and $E:X\to(A;H)$, $F:Y\to(B;H)$ be operator-valued information channels with commuting ranges. There exists a unique operator-valued information channel $E\cdot F:X\times Y\mapsto(A\times B;H)$, such that 
\begin{equation}\label{eq_EcdotF}
(E\cdot F)(\delta\times\theta|x,y) = 
E(\delta|x)F(\theta|y),
\hspace{0.3cm} \delta\in\mathfrak{A}_A, 
\theta\in\mathfrak{A}_B, \ x\in X, y\in Y.
\end{equation}
Moreover if $E$ and $F$ are projection-valued information channels, 
so is $E\cdot F$. 
\end{theorem}
 
\begin{proof}
By Theorem \ref{qpmprod}, there exists 
a family $\left((E\cdot F)(\cdot|x,y)\right)_{(x,y)\in X\times Y}$ of 
quantum probability measures, satisfying (\ref{eq_EcdotF}). 
Using the identification between 
quantum probability measures and unital completely positive maps
from Lemma \ref{l_ext}, 
it suffices, by Remark \ref{r_measfamm}, 
to show that the maps
$\phi_{E\cdot F}^{(x,y)} := \phi_{(E\cdot F)(\cdot|x,y)}$, $(x,y)\in X\times Y$, 
form a jointly measurable family for the Borel $\sigma$-algebra 
$\frak{A}_{X\times Y}$ of $X\times Y$. 
It suffices to show that, if 
$f\in C(X)$, $g\in C(Y)$ and $\xi\in H$, then the function 
\begin{equation}\label{eq_famime}
(x,y)\mapsto\langle\phi_{E\cdot F}^{(x,y)}(f\otimes g)\xi,\xi \rangle
\end{equation}
is measurable with respect to the product $\sigma$-algebra $\frak{A}_X\otimes\frak{A}_Y$. 
As $ \phi_{E\cdot F}^{(x,y)}(f\otimes g) = \phi_E^x(f)\phi_F^y(g)$, 
the function (\ref{eq_famime}) coincides with the function 
$(x,y)\mapsto \langle\phi_E^x(f)\xi,\phi_F^y(g)^*\xi\rangle$.
Letting $(e_k)_{k\in \bb{N}}$ be an orthonormal basis of $H$, we have that 
$$\langle\phi_E^x(f)\xi,\phi_F^y(g)^*\xi\rangle = 
\sum_{k=1}^{\infty} \langle\phi_E^x(f)\xi,e_k\rangle 
\langle \phi_F^y(g)e_k,\xi\rangle,$$
and the joint measurability follows.

Now assume that both $E$ and $F$ are projection-valued information channels. By Corollary \ref{prod_qclass}, we have that $(E\cdot F)(\cdot|x,y)$ is a projection-valued measure for any $(x,y)\in X\times Y$. By the same argument in the case of operator-valued information channels we have that the function 
$$(x,y)\mapsto \langle E\cdot F(\gamma|x,y)\xi,\xi\rangle,\hspace{0.3cm} \gamma\in\mathfrak{A}_A\otimes \mathfrak{A}_B,\xi\in H,$$
is jointly measurable.
\end{proof}

Using a straightforward modification of the proof of Corollary 
\ref{prod_qclass}, Theorem \ref{cqc_prod} implies
the following corollary.

\begin{corollary}\label{q_prod}
Let $A,B,X,Y$ be compact Hausdorff spaces, 
$H$ and $K$ be separable Hilbert spaces, 
and $E:X\to(A;H)$ and 
$F:Y\to(B;K)$ be operator-valued information channels. 
There exists an operator-valued information channel $E\otimes F:X\times Y\mapsto(A\times B,H\otimes K)$, such that
\begin{equation}\label{eq_EotFxy}
(E \otimes F)(\delta\times \theta|x,y) = E(\delta|x)\otimes F(\theta|y), 
\hspace{0.3cm} \delta\in\mathfrak{A}_A, 
\theta\in\mathfrak{A}_B, \ x\in X, y\in Y.
\end{equation}
Moreover, if $E$ and $F$ are projection-valued information channels, 
then so is $E\otimes F$. 
\end{corollary}

\begin{remark}\label{r_muprod}
\rm 
Let $\mu$ (resp. $\nu$) be a Borel probability measure on $X$ (resp. $Y$).
The proof of Theorem \ref{cqc_prod} shows that the following are true:
\begin{itemize}
\item[(i)]
Let $E$ (resp. $F$) be a
$\cl B(H)$-valued $\mu$-information channel 
(resp. $\cl B(H)$--valued $\nu$-information channel). 
Suppose that $E$ and $F$ have commuting ranges. Then there exists an
$\cl B(H)$--valued $\mu\times\nu$-infor\-mation channel $E\cdot F$, for which the equalities (\ref{eq_EcdotF}) hold. 

\item[(ii)]
Let $E$ (resp. $F$) be a $\cl B(H)$-valued $\mu$-information channel 
(resp. $\cl B(K)$-valued $\nu$-information channel). 
Then there exists an
$\cl B(H\otimes K)$-valued $\mu\times\nu$-information channel 
$E\otimes F$, for which the equalities 
(\ref{eq_EotFxy}) hold. 
\end{itemize}
\end{remark}


\section{Dilation results}\label{s_dilres}

In this section, we establish measurable versions of the Stinespring dilation theorem. 
The results will be used subsequently for establishing disambiguation for 
classes of no-signalling correlations, but we believe they may be 
interesting in their own right. 

Given a (Borel) measure $\mu\in M(X)$, we 
recall that $\frak{A}_{\mu}$ denotes the complete $\sigma$-algebra containing $\frak{A}_X$ together with all $\mu$-null subsets of $X$. 
We say that a function defined on $X$ is $\mu$-measurable if it is 
$\frak{A}_{\mu}$-measurable
(note that every Borel function is $\mu$-measurable). 
The $\mu$-measurability of a family of maps from a $C^*$-algebra $\cl A$ to $\cl B(H)$ means the pointwise weak $\mu$-measurability.

\begin{theorem}\label{l_measfree}
Let $X$ be a compact Hausdorff space, $\cl A$ be a 
separable unital C*-algebra, $H$ be a Hilbert space, 
$\mu\in P(X)$ be a non-atomic measure, and 
$\phi_x : \cl A\to \cl B(H)$ be a unital completely positive map, $x\in X$, such that the family 		
$(\phi_x)_{x\in X}$ is $\mu$-measurable.
There exist a Hilbert space $K$, an isometry $V : H\to K$ and a $\mu$-measurable family $(\pi_x)_{x\in X}$ of 
unital *-representations of $\cl A$ on $K$, such that 
\begin{equation}\label{eq_dilasi}
\phi_x(u) = V^*\pi_x(u) V, \ \ \ u\in \cl A, x\in X.
\end{equation}
\end{theorem}

	\begin{proof}
		Let $f:\cl{A}\rightarrow\bb{C}$ be a fixed state. Set $\cl{A}_y=\cl{A}$ for every $y\in X$ and let $\cl B = \ast_{y\in X}\cl A_y$ be the C*-algebraic free product of 
		the family $(\cl A_y)_{y\in X}$, amalgamated over the units. 
		For every $x\in X$, we consider $\cl A_x$ as a unital C*-subalgebra of $\ast_{y\in X}\cl A_y$.
		By \cite{boca} (see also \cite{dk}), there exists a unital completely positive map 
		$\phi : \ast_{y\in X}\cl A_y \to \cl B(H)$ such that, if $u_i\in \cl A_{y_i}$ with $y_1\neq y_2 \neq\cdots \neq y_k$ and 
		$f(u_i) = 0$, $i = 1,\dots,k$, then 
		$$\phi(u_1 u_2\cdots u_k) = \phi_{y_1}(u_1)\phi_{y_2}(u_2)\cdots \phi_{y_k}(u_k).$$
		Following the Stinespring construction (see e.g. \cite[Theorem 1.1.1]{arveson-acta}), 
define a sesqui-linear form on the algebraic tensor product $\cl B\otimes H$ 
by setting
$$\langle u\otimes \xi,v\otimes \eta\rangle = \langle \phi(v^*u)\xi,\eta\rangle, \ \ \ u,v\in \cl B, \ \xi,\eta\in H,$$
let $\cl N$ be its kernel, and let 
$K$ be the Hilbert completion of the quotient $\cl B\otimes H/\cl N$ with respect to 
the induced inner product.
Let $\pi : \cl B\to \cl B(K)$ be the unital *-representation determined by the identities 
$$\pi(v)(u\otimes \xi + \cl N) = vu\otimes \xi  + \cl N, \ \ \ u,v\in \cl B, \ \xi\in H.$$
By the universal property of the C*-algebraic free product, there exist unital *-representations
$\pi_x : \cl A_x\to \cl B(K)$, $x\in X$, such that $\pi = \ast_{x\in X}\pi_x$. 
We will show that the family $(\pi_x)_{x\in X}$ is $\mu$-measurable. 

\smallskip

\noindent {\it Claim.}
Let $u_i\in\cl{A}_{y_i}$ with $y_1\not=y_2\not=\cdots\not=y_k$. 
Then $u_1\cdots u_k$ is a linear combination of the unit $1$ and
 elements of the form $w_1\cdots w_l$, with $l\leq k$, 
where $w_r\in\operatorname{ker}(f)$ and $w_r\in\cl{A}_{t_r}$, $r=1,\dots,l$, with $t_1\not=\cdots\not=t_l$. 		

\smallskip

\noindent {\it Proof of Claim.}
We use induction on the length $k$ of the word $u_1\cdots u_k$. 
The statement is trivial in the case $n = 1$. 
Assume the claim holds for words of length $k-1$, so that 
$$u_1\cdots u_{k-1}
= \sum_{j=1}^l\sum_{\ell=1}^\tau\lambda_{\ell,j} w_1^\ell\cdots w_j^\ell + \mu 1,$$
where $w_r^\ell\in\operatorname{ker}(f)$ and $w_r^\ell\in\cl{A}_{t_r^{\ell}}$, 
$r = 1,\dots,l$ with $t_1\not=\cdots\not=t_l$ and  $\lambda_{\ell,j}, \mu\in\bb{C}$, for $\ell=1,\dots,\tau$, $\tau \leq k-1$. 
Let $u_k\in \cl A_{y_k}$. amd write $u_k=u_k^0+c_k 1$, where $u_k^0\in\ker(f)$ and $c_k\in\mathbb C$.
We have that 
		\begin{align*}
			u_1\cdots u_{k-1}u_k 
   & =\left(\sum_{j=1}^l\sum_{\ell=1}^\tau\lambda_{\ell,j} w_1^\ell\cdots w_j^\ell+\mu1\right)u_k\\
& =\sum_{j=1}^l\sum_{\ell=1}^\tau \lambda_{\ell,j} w_1^\ell\cdots w_j^\ell u_k^0
+ c_k \lambda_{\ell,j} w_1^\ell\cdots w_j^\ell + \mu u_k^0 + c_k\mu1, 
		\end{align*}
Note that the last three terms on the right hand side are of the desired form. 
Fix $j$ and $\ell$. 
If $y_k\not=t_{j}^{\ell}$ then the summand 
$w_1^\ell\cdots w_{j}^\ell u_k^0$ is of the desired form. 
If $y_k = t_{j}^{\ell}$, write 
$w_{j}^\ell u_k^0 = w_{j}^{0,\ell}+c_{j}^\ell1,$ 
where $w_{j}^{0,\ell}\in\operatorname{ker}(f)$ and 
$c_{j}^\ell\in\bb{C}$. 
We then have 		
$$w_1^\ell\cdots w_j^\ell u_k^0
=
w_1^\ell\cdots w_{j-1}^\ell w_{j_0}^{0,\ell} + c_{j_0}^\ell w_1^\ell\cdots w_{j-1}^\ell,$$
which is of the desired form since $w_{j}^{0,\ell}\in\cl{A}_{t_j^{\ell}}$, 
$w_{j-1}^\ell\in\cl{A}_{t_{j-1}}$ and $t_{j}^{\ell}\not=t_{j-1}^{\ell}$. 
Since $j$ and $\ell$ were chosen arbitrarily, the claim is proved.

\smallskip

Fix $u\in \cl A$, 
$u_i\in \cl A_{y_i}$  with $y_1\neq y_2 \neq\cdots \neq y_k$, 
$v_j\in \cl A_{z_j}$ 
with $z_1\neq z_2 \neq\cdots \neq z_m$, and $\xi,\eta\in H$. 
We show that the function 
$$x\mapsto \left\langle \pi_x(u)(u_1\cdots u_k\otimes \xi + \cl N), 
v_1\cdots v_m\otimes \eta + \cl N\right\rangle$$
		is $\mu$-measurable. 
Using the claim, we can assume that 
$u_i\in\operatorname{ker}(f)$ and $v_j\in\operatorname{ker}(f)$, for all $i=1,\dots,k$.
Without loss of generality, we may also assume that $f(u) = 0$. 
		We have 
\begin{eqnarray*}
& & 
\langle \pi_x(u)(u_1\cdots u_k\otimes \xi + \cl N), v_1\cdots v_m\otimes \eta  + \cl N\rangle\\
& = & 
\langle \pi(u)(u_1\cdots u_k\otimes \xi  + \cl N), v_1\cdots v_m\otimes \eta + \cl N\rangle\\
& = & 
\langle \phi((v_1\cdots v_m)^*u(u_1\cdots u_k))\xi,\eta\rangle
=
\langle \phi(v_m^*\cdots v_1^*u u_1\cdots u_k)\xi,\eta\rangle.
\end{eqnarray*}
If $x\not\in \{y_1,z_1\}$ then 
\begin{eqnarray*}
& & 
\langle \phi(v_m^*\cdots v_1^*u u_1\cdots u_k)\xi,\eta\rangle \\
& = & 
\langle\phi_{z_m}(v_m)^*\cdots \phi_{z_1}(v_1)^*\phi_x(u) \phi_{y_1}(u_1)\cdots\phi_{y_k}(u_k)\xi,\eta\rangle\\
& = & 
\langle\phi_x(u) \phi_{y_1}(u_1)\cdots\phi_{y_k}(u_k)\xi,\phi_{z_1}(v_1)\cdots\phi_{z_m}(v_m)\eta\rangle.
\end{eqnarray*}
Since the measure $\mu$ is continuous, $\nu(\{y_1,z_1\}) = 0$. 
It follows that the function
\begin{equation}\label{eq_xmaps}
x\mapsto 
\langle \pi_x(u)(u_1\cdots u_k\otimes \xi + \cl N), v_1\cdots v_m\otimes \eta  + \cl N\rangle
\end{equation}
coincides $\mu$-almost everywhere with a $\mu$-measurable function, and is hence 
$\mu$-measurable itself. 

Now let $\tilde{\xi}, \tilde{\eta} \in K$ be arbitrary. 
Let $(\xi_n)_{n\in \bb{N}}$ and $(\eta_n)_{n\in \bb{N}}$ be sequences 
of vectors in $K$, each being a 
finite sum of vectors of the form $u_1\cdots u_k\otimes \xi + \cl N$. 
Then the function 
$x\mapsto \langle \pi_x(u)\tilde{\xi}, \tilde{\eta}\rangle$
is a uniform limit of functions of the form (\ref{eq_xmaps}), and is hence 
$\mu$-measurable. 
\end{proof}

\begin{remark}\label{r_discrete}
\rm 
Suppose that $X$ is a set, $\cl A$ is a 
separable unital C*-algebra, $H$ is a Hilbert space, and
$\phi_x : \cl A\to \cl B(H)$ is a unital completely positive map, $x\in X$.
Then the proof of Theorem \ref{l_measfree} shows that 
there exist a Hilbert space $K$, an isometry $V : H\to K$ and a 
family $(\pi_x)_{x\in X}$ of 
unital *-representations of $\cl A$ on $K$, which fulfills the 
dilation identities (\ref{eq_dilasi}). 

Note that if $\mu$ is totally atomic then the family of $*$-representations $(\pi_x)_{x\in X}$, constructed in the proof, is automatically measurable, and the statement of the theorem holds true in this case as well.  
\end{remark}

\begin{remark}
\rm 
Let $\cl{A}$ be a unital C$^*$-algebra $\cl{A}$, following the same arguments as in Theorem \ref{th_Borel}, we obtain a one-to-one 
correspondence between $\mu$-measurable families of unital completely positive maps 
$\phi_x:\cl{A}\rightarrow\cl{B}(H)$, $x\in X$ and unital completely positive maps $\Phi:\cl{A}\rightarrow \frak{B}_\mu(X;H)$, 
satisfying the equation 
$$\Phi(u)(x)=\phi_x(u), \hspace{0.2cm}u\in \cl{A}, x\in X.$$
\end{remark}

The next theorem shows that a joint measurable dilation 
of a pair of measurable families of unital completely positive maps 
is also possible.

\begin{theorem}\label{disambiguation}
Let $X$ (resp. $Y$) be a compact Hausdorff spaces, $\mu\in P(X)$
(resp. $\nu\in P(Y)$) be a continuous or totally atomic measure, $\cl A$ (resp. $\cl B$)
be a unital separable C*-algebra, 
$\phi_x:\cl A\rightarrow\cl{B}(H)$, $x\in X$
(resp. $\psi_y:\cl B\rightarrow\cl{B}(H)$ $y\in Y$) be a $\mu$-measurable 
(resp. $\nu$-measurable) family of unital completely positive maps, 
such that $\phi_x$ and $\psi_y$ have commuting ranges for all $(x,y)\in X\times Y$. There exists a Hilbert space $K$, a $\mu$-measurable family $(\pi_x)_{x\in X}$ of $^*$-representations of $\cl A$ on $K$, a $\nu$-measurable family $(\rho_y)_{y\in Y}$ of $^*$-representations of $\cl B$ on $K$  with commuting ranges for all $(x,y)\in X\times Y$, and an isometry $V : H\rightarrow K$, such that
$$(\phi_x\cdot\psi_y)(w)
= V^*(\pi_x\cdot\rho_y)(w)V,\hspace{0.3cm} w\in\cl{A}\otimes_{\max}\cl{B}, \hspace{0.2cm}(x,y)\in X\times Y.$$
	\end{theorem} 
 
	\begin{proof}
Let $\phi=\ast_{x\in X}\phi_x$ and $\psi=\ast_{y\in Y}\psi_y$ the 
unital completely positive maps obtained via Boca's Theorem, 
acting on $\tilde{\cl{A}}= \ast_{x\in X}\cl{A}_x$ and $\tilde{\cl{B}} = \ast_{y\in Y}\cl{B}_y$, respectively; here we have set $\cl{A}_x = \cl A$, $x\in X$, 
and $\cl{B}_y = \cl B$, $y\in Y$. 
We show that $\phi$ and $\psi$ have commuting ranges. Let $u_i\in\cl{A}_{y_i}$ for $i=1,\dots,k$ with $y_1\neq y_2\neq\dots\neq y_k$, 
and $v_j\in\cl{A}_{z_j}$ for $j=1,\dots,m$ with $z_1\neq z_2\neq\dots\neq z_m$. For $\xi,\eta\in H$, we have
		\begin{align*}
			\langle \phi(u_1\dots u_k)\psi(v_1\dots v_m)\xi,\eta\rangle & =\langle\phi_{y_1}(u_1)\cdots\phi_{y_k}(u_k)\psi_{z_1}(v_1)\cdots\psi_{z_m}(v_m)\xi,\eta\rangle \\
			& =\langle \psi_{z_1}(v_1)\cdots\psi_{z_m}(v_m)\phi_{y_1}(u_1)\cdots\phi_{y_m}(u_m)\xi,\eta\rangle\\
			& =\langle \psi(v_1\dots v_m)\phi(u_1\dots u_k)\xi,\eta\rangle.
		\end{align*}
Let  $\sigma : \tilde{\cl{A}}\otimes_{\max} \tilde{\cl{B}}\rightarrow\cl{B}(H)$
be the unital completely positive map, such that 	
		\begin{align*}
			\sigma(a\otimes b)=\phi(a)\psi(b),\hspace{0.3cm}
a\in \tilde{\cl{A}}, b\in\tilde{\cl{B}}.
		\end{align*}
After applying Stinespring's Dilation Theorem to $\sigma$, 
we obtain a Hilbert space $K$ as a completion of a quotient of 
the linear space 
$(\tilde{\cl A}\otimes_{\max}\tilde{\cl B})\odot H$ by a subspace $\cl N$, 
a $^*$-representation $\tilde{\sigma}$ of 
$\tilde{\cl{A}}\otimes_{\max}\tilde{\cl B}$ on $K$ and an isometry 
$V\in\cl{B}(H,K)$, such that
\begin{align*}
\sigma(t)=V^*\tilde{\sigma}(t)V,\hspace{0.3cm} t\in \tilde{\cl{A}}\otimes_{\max}\tilde{\cl{B}}.
\end{align*}
Let $\pi : \tilde{\cl{A}}\rightarrow\cl{B}(K)$ and 
$\rho : \tilde{\cl{B}}\rightarrow\cl{B}(K)$ be the $^*$-homomorphisms, given by 
		$\pi(a) = \tilde{\sigma}(a\otimes1)$, $a\in \tilde{\cl{A}}$, 
  and $\rho(b)=\tilde{\sigma}(1\otimes b)$, $b\in \tilde{\cl{B}}$. 
Further, let $\pi_x$, $x\in X$ (resp. $\rho_y$, $y\in Y$) be the $*$-representation 
of $\cl A$ (resp. $\cl B$) on $K$, such that $\pi = \ast_{x\in X} \pi_x$
(resp. $\rho = \ast_{y\in Y} \rho_y$).
Now similarly to the proof of Theorem \ref{l_measfree} one can verify the 
$\mu$-measurability of the family $(\pi_x)_{x\in X}$  if $\mu$ is continuous; 
by symmetry, we also obtain the $\nu$-measurability of the family the family 
$(\rho_y)_{y\in Y}$, if $\nu$ is continuous; for totally atomic measures the measurability is automatic. 
	\end{proof}

Another application of Theorem \ref{l_measfree} is the following measurable version of Naimark's Dilation Theorem.

\begin{corollary}\label{c_disformu}
Let $\mu\in P(X)$ (resp. $\nu\in P(Y)$) be a continuous measure and 
$E$ (resp. $F$) be 
a $\cl B(H)$-valued $\mu$-information channel over $(X,A)$
(resp. $\cl B(H)$-valued $\nu$-information channel over $(Y,B)$). 
Assume that $E$ and $F$ have commuting ranges. 
Then there exist a Hilbert space $K$, an isometry 
$V : H\rightarrow K$ and 
a $\cl B(K)$-valued $\mu$-information channel $P$ over $(X,A)$
(resp. $\cl B(K)$-valued $\nu$-information channel $Q$ over $(Y,B)$)
that have commuting ranges, whose values are projections, 
such that 
$$E(\alpha|x)F(\beta|y) = V^*P(\alpha|x)Q(\beta|y)V, \hspace{0.3cm}\alpha\in\frak{A}_A, \beta\in \frak{A}_B, x\in X, y\in Y.$$
\end{corollary}
\begin{proof}
    Let $E$ be a $\cl{B}(H)$-valued $\mu$-information channel over $(X,A)$ and $F$ be a $\cl{B}(H)$-valued $\nu$-information channel over $(Y,B)$ such that $(E,F)$ is a commuting pair. By Remark \ref{r_measfamm} we obtain a $\mu$-measurable family $\phi_{E_x}:C(A)\rightarrow\cl{B}(H)$, $x\in X$ of unital completely positive maps and a $\nu$-measurable family $\phi_{F_y}:C(A)\rightarrow\cl{B}(H)$, $y\in Y$ of unital completely positive maps such that $\phi_{E_x}$ and $\phi_{F_y}$ have commuting ranges for every $(x,y)\in X\times Y$. By Theorem \ref{disambiguation} we obtain a Hilbert space $K$, a $\mu$-measurable family $(\pi_x)_{x\in X}$ of $^*$-representations of $C(A)$ on $K$, a $\nu$-measurable family $(\rho_y)_{y\in Y}$ of $^*$-representations of $C(B)$ on $K$  with commuting ranges for all $(x,y)\in X\times Y$, and an isometry $V : H\rightarrow K$, such that
$$(\phi_x\cdot\psi_y)(w)
= V^*(\pi_x\cdot\rho_y)(w)V,\hspace{0.3cm} w\in C(A)\otimes_{\max}C(B), \hspace{0.2cm}(x,y)\in X\times Y.$$
 To complete the proof apply the approximation arguments of Lemma \ref{l_ext} on $(\pi_x\cdot\rho_y)_{(x,y)\in X\times Y}$.
\end{proof}

We note that 
the dilation (\ref{eq_dilasi}) in Theorem \ref{l_measfree} holds at any point $x\in X$. 
At the same time, we have imposed a restriction on the measure $\mu$ (assumed to 
be continuous), and the dilating Hilbert space $K$, in the case where $X$ is not 
countable, is non-separable. 
In the next statement we offer a different result, where the dilations 
hold up to a negligible set, but there is no continuity restriction on the 
measure $\mu$, and the dilating Hilbert space can be chosen to be separable. 
If $X$ is a compact Hausdorff space, $\mu$ is a probability measure on $X$, 
and $H$ is a separable Hilbert space, we 
denote by $L^\infty (X,\mu, \cl B(H))$ the space of all 
(equivalence classes of) essentially
norm-bounded functions $F : X\to \cl B(H)$, such that the scalar-valued 
functions $x\to \langle F(x)\xi,\eta\rangle$ are $\mu$-measurable 
for all $\xi,\eta\in H$. 

\begin{theorem}\label{l_measfree_comm}
Let $X$ be a compact Hausdorff space, $\mu$ be a Borel probability measure on $X$, 
$A$ be a compact Hausdorff space, $H$ be a separable Hilbert space, 
and $\phi: C(A)\to 
L^\infty (X,\mu, \cl B(H))$ be a unital completely positive map.
Then there exist a separable 
  Hilbert space $K$, an isometry $V : H\to K$ and a $\mu$-measurable family $(\pi_x)_{x\in X}$ of 
		unital *-representations of $C(A)$ on $K$, such that if $u\in C(A)$ then
		$$\phi(u)(x) = V^*\pi_x(u) V, \ \ \text{ $\mu$-almost everywhere.}$$
	\end{theorem}

\begin{proof}
Fix countable dense subsets $\cl A\subseteq C(A)$ and $\cl E\subseteq H$  in $C(A)$ and $H$, respectively,
with $1\in \cl A$. 
For $c_i, d_i\in C(A)$, $\xi_i$, $\eta_i\in H$,  
$i = 1,\dots, k$, set
$$\left\langle \sum_{i=1}^k c_i\otimes \xi_i,\sum_{j=1}^k d_j\otimes\eta_j
\right\rangle(x)
:=\sum_{i,j=1}^k \langle\phi(d_j^*c_i)(x)\xi_i,\eta_j\rangle.$$
We have that the function  
$$x\mapsto 
\left\langle \sum_{i=1}^k c_i\otimes \xi_i,\sum_{j=1}^k d_j\otimes\eta_j\right\rangle(x)$$
is $\mu$-measurable; for $\tilde{\xi} = \sum_{i=1}^k c_i\otimes \xi_i \in C(A)\otimes H$, let 
$X_{\tilde{\xi}}$ be a subset of $X$ such that $\mu\left(X\setminus X_{\tilde{\xi}}\right) = 0$ and 
\begin{equation}\label{eq_pdx}
\left\langle \tilde{\xi},\tilde{\xi}\right\rangle(x) = 
\left\langle \phi^{(k)}\left((c_j^*c_i)_{i,j}\right)(x)\xi,\xi\right\rangle\geq 0, \ \ \ x\in X_{\tilde{\xi}}
\end{equation}
(here we have set $\xi=(\xi_1,\ldots,\xi_k)\in H^{(k)}$). 
Let $X_0 = \cap\{X_{\tilde{\xi}} : \tilde{\xi}\in [\cl A\odot \cl E]_{\bb{Q}}\}$, 
where $[\cl A\odot \cl E]_{\bb{Q}}$ is the linear span of $\cl A\odot \cl E$ over the field $\bb{Q} + i\bb{Q}$ of complex rationals;
clearly, $X_0$ is $\mu$-measurable and $\mu\left(X\setminus X_0\right) = 0$. 
Let $x\in X_0$ and $\tilde{\xi} = \sum_{i=1}^k \lambda_k c_k\otimes \xi_k \in [\cl A]\odot H$; then there exists a sequence 
$(\tilde{\xi}_n)_{n\in \bb{N}} \subseteq [\cl A\odot \cl E]_{\bb{Q}}$, 
with 
$\tilde{\xi}_n = \sum_{i=1}^k \lambda_{k,n} c_{k}\otimes \xi_{k,n}$, where $\lambda_{k,n}\in \bb{Q}$ and 
$\xi_{k,n}\in \cl E$, such that $\lambda_{k,n} \to_{n\to\infty} \lambda_k$ and $\xi_{k,n}\to_{n\to\infty} \xi_k$. 
Equation (\ref{eq_pdx}) now shows that 
$\left\langle\tilde{\xi},\tilde{\xi} \right\rangle(x) \geq 0$ for each $x\in X_0$. 
We can further assume, after deleting a further set of measure zero if necessary, that $\Phi(1)(x)=1$ for all $x\in X_0$. 

Let 
$$I_x = \left\{a\in [\cl A]\odot H: \langle a,a\rangle (x)=0 \right\}, \ \ \ x\in X_0.$$ 
By the Cauchy-Schwartz inequality, 
$$I_x = \left\{a\in[ \cl A]\odot H:\langle a,b\rangle(x)=0, 
\mbox{ for all } b\in[\cl A]\odot H\right\},$$  
and hence $I_x$ is a subspace of $[\cl A]\odot H$.
	
	Let $K_x$  be  the completion of $[\cl A]\odot H/I_x$ with respect to the induced inner product and set  $(a\otimes\xi)(x)=a\otimes\xi+I_x$, $x\in X_0$, $a\in \cl A$, $\xi\in H$.
Then 

$$\cl F:=\{((a\otimes\xi)(x))_{x\in X_0}\in \textstyle\prod_{x\in X_0}K_{x}: a\in\cl A, \xi\in\cl E\}$$ is a countable subset such that,  
	for each $x\in X_0$, 
 the set $\{(a\otimes\xi)(x): a\in \cl A, \xi\in\cl E\}$ is total in $K_x$.
As the function 
$x\mapsto \langle (a\otimes\xi)(x), (b\otimes\eta)(x)\rangle=\langle\phi(b^*a)(x)\xi,\eta\rangle$ is $\mu$-measurable, by \cite[Lemma 8.10]{takesaki}, 
	\begin{eqnarray*}
\frak M 
& := & 
\{\zeta\in \textstyle\prod_{x\in X_0}K_x: x\mapsto\langle\zeta(x),(a\otimes\xi)(x)\rangle \\
& & 
\hspace{2cm} \text{ is $\mu$-measurable for all }a\in\cl A,\xi\in\cl E\}
	\end{eqnarray*}
	 is a measurable vector field (see \cite[Definition 8.9]{takesaki}). 
	 
Let 
$$\cl H = \left\{\zeta\in\frak M : \|\zeta\| = \left(\int_X\|\zeta(x)\|^2d\mu(x)\right)^{1/2} < \infty\right\}$$ 
be the Hilbert space, associated with the field $\frak M$, which we also denote, 
as usual, by $\int_X^{\oplus}K_xd\mu(x)$. 
Let 
$\zeta=((a\otimes\xi)(x))_x$, $a\in\cl A$, $\xi\in \cl E$. We have that 
$$\|\zeta\|^2=\int_X\langle\phi(a^*a)(x)\xi,\xi\rangle d\mu(x)\leq\|\|\phi(a^*a)(\cdot)\|\|_\infty\|\xi\|^2<\infty,$$
and hence $\zeta \in \cl H$.
The countable set $\cl F$ is a fundamental sequence of the vector field $\frak M$ as defined in \cite[Definition 8.9]{takesaki}.
By \cite[Lemma 8.12]{takesaki}, the function 
$x\mapsto n(x) := \text{dim} (K_x)$ is $\mu$-measurable  and hence 
the set $\{x\in X: n(x)=k\}$ is $\mu$-measurable for any $k=0,1,2,\ldots, +\infty$. 

Let $\{\xi_k\}_{k\in \bb{N}}$ be an orthonormal basis of $H$, which we assume without loss of 
generality, to be contained in $\cl E$. 
We observe that $\langle(1\otimes\xi_k)(x), (1\otimes\xi_m)(x)\rangle=\langle\phi(1)(x)\xi_k,\xi_m\rangle=\delta_{k,m}$, $x\in X_0$. 
Then arguments as in \cite[Lemma 8.12]{takesaki} give that we can find a fundamental sequence of measurable vector fields 
$\{\zeta_n\}_n=\{\tilde\zeta_k\}_{k=1}^{r}\cup\{\tilde{\tilde\zeta}_m\}_{m=1}^{\infty}$, where $r=\text{dim}(H)$, such that 
$\tilde\zeta_k=((1\otimes\xi_k)(x))_x$, and such that 
$\{\tilde\zeta_k(x)\}_{k=1}^r\cup
\{\tilde{\tilde\zeta}_m(x)\}_{m = 1}^{m(x)}$ 
is an orthonormal basis of $K_x$ for all $x\in X_0$ where, 
setting $H_x=\overline{\text{span}}\{\tilde\zeta_k(x)\}_{k=1}^{r}$, we 
have $m(x)=\text{dim} (H_x^\perp)$ and $\tilde{\tilde\zeta}_{m(x)+k}(x)=0$, 
$k\in \bb{N}$, if $m(x)<\infty$. We note that the construction of $\{\zeta_n(x)\}$  is nothing else than the Gram-Schmidt orthogonalization of $\cl F$ and hence each $\zeta_n(x)$ is a linear combination of elements in $\cl F$ with measurable coefficients: $\zeta_n(x)=\sum_k e_k(x)s_k(x)$, where $s_k\in\cl F$ and $x\mapsto e_k(x)\in\mathbb C$ is measurable.

	 Let $K$ be a separable Hilbert space  with $\{\varepsilon_n\}_n=\{\tilde\varepsilon_k\}_{k=1}^r\cup\{\tilde{\tilde\varepsilon}_m\}_{m=1}^{\infty}$ as its orthonormal basis, and define 
	 $$U(x)\zeta=\sum_{k=1}^r\langle\zeta,\tilde\zeta_k(x)\rangle\tilde\varepsilon_k
  + \sum_{m=1}^{m(x)} \langle\zeta,\tilde{\tilde\zeta}_m(x)\rangle\tilde{\tilde\varepsilon}_m , 
  \ \ \ \zeta\in K_x.$$
	 We have that the function 
  $x\mapsto U(x)\zeta(x)$ is measurable if so is the function $x\mapsto \zeta(x)$. 
	 
	 Let $E_k=\{x\in X_0:m(x)=k\}$ and $H_k=\ell^2(\mathbb N_{r+k})=\overline{\text{span}}(\{\tilde\varepsilon_k\}_{k=1}^r\cup\{\tilde{\tilde\varepsilon}_m\}_{m=1}^k)$, $\cl H_k=\int_{E_k}^{\oplus}K_xd\mu(x)$, and $U_k=\int_{E_k}^\oplus U(x)d\mu(x)$. We have that $\{E_k\}_k$ is a family of pairwise disjoint  measurable sets, $\cl H=\oplus_k\cl H_k$  and each $U_k: \cl H_k\to L^2(E_k,\mu)\otimes H_k$ is a unitary operator.  
  We may assume that $\cl A$ is closed under multiplication and consider for $x\in X_0$, $\pi_x: [\cl A]\to \cl B(K_x)$, $\pi_x(a)\zeta(x)=\sum_{i=1}^m (aa_i\otimes\xi_i)(x)$, for $\zeta(x)=\sum_{i=1}^m (a_i\otimes\xi_i)(x)$, $a\in[\cl A]$,  $a_i\in\cl A$, $\xi_i\in H$, $i = 1,\dots,m$. As
	 \begin{eqnarray*}
	&& \langle\pi_x(a)\zeta(x),\pi_x(a)\zeta(x)\rangle=\sum_{i,j=1}^m\langle\phi(a_j^*a^*aa_i)(x)\xi_i,\xi_j\rangle \\&&\leq\|a\|^2\sum_{i,j=1}^{m}\langle\phi(a_j^*a_i)(x)\xi_i,\xi_j\rangle =\|a\|^2\langle\zeta(x),\zeta(x)\rangle =\|a\|^2\|\zeta(x)\|^2
	 \end{eqnarray*}
	 and $\{(\sum_{i=1}^ma_i\otimes\xi_i)(x): a_i\in \cl A, \xi_i\in H\}$ is dense in $K_x$, the map 
  $\pi_x$ can be extended to a well-defined representation  of $C(A)$ on $K_x$. 
  
  We observe next that $x\in E_k\mapsto \tilde\pi_x(a):=U(x)\pi_x(a)U(x)^*\in \cl B(H_k)$, $a\in C(A)$, is weakly measurable for each $k$. In fact, 
  letting $\cl F(x) = \{\tilde{\xi}(x) : \tilde{\xi}\in \cl F\}$, for $x\in X_0$, 
  for $a\in[\cl A]$, we have that $\pi_x(a)\cl F(x)\subseteq\cl F(x)$, and $U(x)^*\varepsilon_n=\zeta_n(x)$, so that by the construction of $\zeta_n$, we obtain that $x\mapsto \langle U(x)\pi_x(a)U(x)^*\varepsilon_n,\varepsilon_m\rangle$ is measurable. That $x\mapsto\langle U(x)\pi_x(a)U^*(x)\xi,\eta\rangle$ is measurable for all $a\in C(A)$ and $\xi$, $\eta\in H_k$ follows from the separability. 
	 
	 
	 Let $V : L^2(X,\mu)\otimes H = \bigoplus_k L^2(E_k,\mu)\otimes H\to\cl H=\bigoplus_k\cl H_k$  be given by $V(\xi\otimes\eta)(x)=\xi(x)(1\otimes\eta)(x)$, $\xi\in L^2(X,\mu)$, $\eta\in H$,
	 and $U=\bigoplus_kU_k:  \oplus_k\cl H_k\to \bigoplus_kL^2(E_k,\mu,H_k)$. Then
	 $$(UV(\xi\otimes\eta))(x)=\sum_{k=1}^r\xi(x)\langle (1\otimes\eta)(x),\tilde\zeta_k(x)\rangle\tilde\varepsilon_k\in L^2(E_k,\mu, H_k)$$ if $\xi\in L^2(E_k,\mu)$, $\eta\in H$.
	 
  Extend
	 each $\tilde\pi_x(a)$, $x\in E_k$, to a representation on $K$ by letting for a fixed $\alpha\in A$, $\tilde\pi_x(a)\xi= a(\alpha)\xi$ for any $\xi\in H_k^\perp$ and set $\tilde\pi(a)=\int^\oplus_X \tilde \pi_x(a)d\mu(x)$, as a representation on $L^2(X,\mu)\otimes K$. 
	 For $\xi,\tilde{\xi}\in L^2(X,\mu)$ and $\eta,\tilde{\eta}\in H$, we have 
	 \begin{eqnarray*}
	 \langle V^*U^*\tilde\pi(a)UV(\xi\otimes\eta),\tilde\xi\otimes\tilde\eta\rangle
  & = & 
  \int_X \langle \pi_x(a)V(\xi\otimes\eta)(x),V(\tilde\xi\otimes\tilde\eta)(x)\rangle d\mu(x)\\
  & = & 
	 \int_X\langle\phi(a)(x)\eta,\tilde\eta\rangle\xi(x)\overline{\tilde\xi(x)}d\mu(x).
	 \end{eqnarray*}
	 On the other hand,
	 \begin{eqnarray*}
	 &&\langle V^*U^*\tilde\pi(a)UV(\xi\otimes \eta), \tilde\xi\otimes\tilde\eta\rangle\\&&= \int_X\langle \tilde\pi_x(a)UV(\xi\otimes\eta)(x), UV(\tilde\xi\otimes\tilde\eta)(x)\rangle d\mu(x)\\&&=\int_X\langle\tilde\pi_x(a)\xi(x)\sum_{k=1}^r\langle\eta,\xi_k\rangle\tilde\varepsilon_k,\tilde\xi(x)\sum_{k=1}^r\langle\tilde\eta,\xi_k\rangle\tilde\varepsilon_k\rangle d\mu(x)\\&&=\int_X\langle\tilde\pi_x(a)\tilde V\eta, \tilde V\tilde\eta\rangle \xi(x)\overline{\tilde\xi(x)}d\mu(x),
	 \end{eqnarray*}
	 where $\tilde V: H\to K$ is the operator, given by $\eta\mapsto \sum_{k=1}^r\langle\eta,\xi_k\rangle\tilde\varepsilon_k$,
implying that $\tilde V^*\tilde\pi_x(a)\tilde V=\phi(a)(x)$, $\mu$ almost everywhere. As $\phi$ is unital, $\tilde V$ is an isometry. 
\end{proof}

\begin{remark}\rm 
    Theorem \ref{l_measfree_comm} holds true for general unital separable 
    $C^*$-algebras in the place of $C(A)$. Indeed, by enlarging the space $K$ so that $H_k^\perp$ is infinite-dimensional, we can extend $\tilde\pi_x$, $x\in E_k$, to a representation on $K$ by letting $\tilde\pi_x|_{H_k^\perp}$ to be a fixed representation of the $C^*$-algebra.  
\end{remark}


\section{Measurable no-signalling correlations}\label{s_mnsc}

In this section, we define measurable no-signalling correlations, 
distinguishing three types thereof, define natural products between 
correlations, and 
show that the latter are type-preserving. 
Let $X$, $Y$, $A$ and $B$ be compact Hausdorff spaces. 
An information channel
$$p = (p(\cdot,\cdot|x,y))_{(x,y)\in X\times Y} : X\times Y \mapsto A\times B$$ 
is called a \emph{measurable no-signalling (NS) correlation} over the quadruple 
$(X,Y,A,B)$ if
\begin{equation}\label{eq_ns1}
p(\alpha\times B | x,y) = p(\alpha\times B | x,y') \ \mbox{ for all }
x\in X, y,y'\in Y \mbox{ and } \alpha\in \frak{A}_A,
\end{equation}
and 
\begin{equation}\label{eq_ns2}
p(A\times \beta | x,y) = p(A\times \beta | x',y)  \ \mbox{ for all } x,x'\in X, y\in Y  \ \mbox{ and } \beta\in \frak{A}_B.
\end{equation}
We denote the set of all no-signalling (measurable) correlations by 
$\cl C_{\rm ns} = \cl C_{\rm ns}(X,Y,A,B)$. 
When understood from the context, we skip the adjective \lq measurable'
and refer to these objects as NS correlations.

Suppose that $\mu$ (resp. $\nu$) is a Borel probability measure on $X$ (resp. $Y$). 
We define \emph{no-signalling $\mu,\nu$-correlations} 
by imposing conditions (\ref{eq_ns1}) and (\ref{eq_ns2}), but requiring
that $p$ be a $\mu\times\nu$-information channel.

\begin{remark}
\rm 
Let $p$ be an NS correlation over the quadruple 
$(X,Y,A,B)$. Then the marginal
$p_A(\cdot | x)$ (resp. $p_B(\cdot | y)$) of $p(\cdot,\cdot|x,y)$ is a well-defined probability measure on 
$A$ (resp. $B$), and the family 
$(p_A(\cdot | x))_{x\in X}$ (resp. $(p_B(\cdot | y))_{y\in Y}$) is an information channel from $X$ (resp. $Y$) to $A$ (resp. $B$). 
\end{remark}

We next define some subclasses of $\cl C_{\rm ns}$. 
If $H$ is a Hilbert space, $E : X\mapsto (A;H)$ and $\xi\in H$ is a unit vector, 
we let $p_{E,\xi} = (p_{E,\xi}(\cdot|x))_{x\in X}$ be the family of 
probability measures, given by 
$$p_{E,\xi}(\alpha|x) = \langle E(\alpha|x)\xi,\xi\rangle, \ \ \ \alpha\in \frak{A}_A, 
x\in X.$$
Trivially, $p_{E,\xi}$ is a classical information channel from $X$ to $A$.

\begin{definition}\label{d_subclasses}
Let $p = (p(\cdot | x,y))_{x,y}$ be a measurable NS correlation over $(X,Y,A,B)$. 
We call $p$ 
\begin{itemize}
\item[(i)]
\emph{local} if there exist $n\in \bb{N}$ and information channels 
$p_i^{(1)} : X\to A$ and $p_i^{(2)} : Y\to B$ and 
scalars $\lambda_i \in [0,1]$, $i = 1,\dots,n$,
such that $\sum_{i=1}^n\lambda_i = 1$ and 
$p = \sum_{i=1}^n \lambda_i p_i^{(1)}\otimes p_i^{(2)}$;
\item[(ii)]
\emph{quantum spatial} 
(resp. \emph{projective quantum spatial}) if there exist
Hilbert spaces $H$ and $K$, 
operator-valued (resp. projection-valued) 
channels $E : X\to (A;H)$ and $F : Y\to (B;K)$ and a unit
vector $\xi\in H\otimes K$, such that $p = p_{E\otimes F,\xi}$;

\item[(iii)]
\emph{quantum commuting} 
(resp. \emph{projective quantum commuting})
if there exist a Hilbert space $H$, pair 
$(E,F)$ of operator-valued 
(resp. projec\-tion-valued) channels with commuting ranges acting on $H$ and a unit
vector $\xi\in H$, such that $p = p_{E\cdot F,\xi}$.
\end{itemize}
\end{definition}

We denote the classes of all quantum commuting (resp. quantum spatial, local) measurable NS correlations by $\cl C_{\rm qc}$ (resp. $\cl C_{\rm qs}$, $\cl C_{\rm loc}$). 
We write $\cl C_{\rm qc}^{\rm pr}$ (resp. $\cl C_{\rm qs}^{\rm pr}$)
for the class of all projective quantum commuting (resp. projective quantum
spatial)
measurable NS correlations. 

Suppose that $\mu$ (resp. $\nu$) is a Borel probability measure on $X$ (resp. $Y$). 
We define $\mu,\nu$-versions of the correlation classes introduced in 
Definition \ref{d_subclasses}, but utilising 
(operator-valued) $\mu$-information channels and 
(operator-valued) $\nu$-information channels. 
We use the notations $\cl C_{\rm t}(\mu,\nu)$ and 
$\cl C_{\rm t}^{\rm pr}(\mu,\nu)$ for the corresponding 
sets of $\mu,\nu$-correlations of type ${\rm t}$. 
Further, we write $\cl C_{\rm qc,sep}(\mu,\nu)$
(resp. $\cl C_{\rm qs,sep}(\mu,\nu)$) 
for the set of quantum commuting $\mu,\nu$-correlations 
(resp. quantum spatial $\mu,\nu$-correlations) that admit 
a realisation as in Definition \ref{d_subclasses} (iii)
(resp. Definition \ref{d_subclasses} (ii)) with a 
separable Hilbert space $H$ (resp. separable Hilbert spaces 
$H$ and $K$).

\begin{remark}\label{r_coincides}
\rm 
\begin{itemize}
\item[(i)]
Suppose that $X$, $Y$, $A$ and $B$ are finite sets, equipped with 
discrete topology. 
The definitions of the 
classes of local, quantum spatial 
and quantum commuting correlations over 
$(X,Y,A,B)$ are in this case well-known; see e.g. \cite{psstw}. 
It is straightforward to see that, 
the classes $\cl C_{\rm t}$, ${\rm t}\in \{{\rm loc}, {\rm qs}, {\rm qc}\}$, 
defined herein, agree with these (finite) classical correlation classes.

\item[(ii)]
We have that 
\begin{equation}\label{eq_inclu}
\cl C_{\rm loc}\subseteq \cl C_{\rm qs} \subseteq \cl C_{\rm qc} \subseteq \cl C_{\rm ns}.
\end{equation}
Indeed, the third inclusions in (\ref{eq_inclu}) in implicit in Definition \ref{d_subclasses} and straightforward to verify. The second inclusion is 
a direct consequence of the definitions. 
To see the first inclusion, 
assume that $p$ has the form in Definition \ref{d_subclasses} (i). 
Let $E : X\to (A;\bb{C}^n)$ and $F : Y\to (B;\bb{C}^n)$ be given by 
$E(\alpha|x) = (p_i^{(1)}(\alpha|x))_{i=1}^n$ and 
$F(\beta|y) = (p_i^{(2)}(\beta|y))_{i=1}^n$, viewed as diagonal operators acting on 
$\bb{C}^n$, and $\xi = (\sqrt{\lambda_i}\delta_{i,j})_{i,j=1}^n$, viewed as a (unit) vector in 
$\bb{C}^n\otimes \bb{C}^n$. 
A direct verification now shows that $p = p_{E\otimes F,\xi}$, and 
(\ref{eq_inclu}) is established. 
We note, in addition, the inclusions 
\begin{equation}\label{eq_inclu12}
\cl C_{\rm qs}^{\rm pr} \subseteq \cl C_{\rm qc}^{\rm pr}, \ \ 
\cl C_{\rm qs}^{\rm pr}\subseteq \cl C_{\rm qs} \ \mbox{ and } \ 
\cl C_{\rm qc}^{\rm pr}\subseteq \cl C_{\rm qc}.
\end{equation}
\end{itemize}
We remark that inclusions, analogous to (\ref{eq_inclu})
and (\ref{eq_inclu12}), continue to hold for the corresponding classes of 
no-signalling $\mu,\nu$-correlations. 
\end{remark}

In the sequel, given 
sets $A$ and $B$, we sometimes abbreviate $A\times B$ to $AB$, and 
an ordered pair $(a,b)\in A\times B$ to $ab$. 
If $A$ and $B$ are topological spaces, $AB$ is equipped with the product 
topology. 
Further, if $A$, $A'$, $B$ and $B'$ are sets, $L\subseteq AB$ and 
$L'\subseteq A'B'$, we let 
$$L\dot{\times} L' = \left\{((aa',bb')\in AA'\times BB' : (a,b)\in L
\mbox{ and } (a',b')\in L'\right\}.$$
It is clear that, if $L$ and $L'$ are Borel subsets of 
$AB$ and $A'B'$, respectively, 
then $L\dot{\times} L'$ is a Borel subset of $AA'\times BB'$. 

Let $X$, $X'$, $Y$, $Y'$, $A$, $A'$, $B$ and $B'$ be compact Hausdorff spaces, and 
$p: XX'\times YY' \mapsto AA'\times BB'$ be an information channel. 
Fix $x'\in X'$ and $y'\in Y'$ and let 
$p_{x',y'} : X\times Y \mapsto A\times B$ be the family of measures over 
$(X,Y)$, given by 
$$p_{x',y'}(L|x,y) = p(L\dot{\times} A'B' | xx',yy'), \ \ \ x\in X, y\in Y, 
\ L\in \frak{A}_{AB}.$$
Then $p_{x',y'}$ is an information channel $XY\mapsto AB$; 
indeed, since the function 
$$(xx'',yy'')\mapsto p(L\dot{\times} A'B' | xx'',yy'')$$
is Borel, so is the function 
$(x,y)\mapsto p(L\dot{\times} A'B' | xx',yy')$. 

\begin{proposition}\label{l_reduction}
Let $X$, $X'$, $Y$, $Y'$, $A$, $A'$, $B$ and $B'$ be compact Hausdorff spaces, and 
$p: XX'\times YY' \mapsto AA'\times BB'$ be a measurable no-signalling correlation. 
Then the information channel
$p_{x',y'} : X\times Y \mapsto A\times B$ 
is a measurable no-signalling correlation, $x'\in X'$, $y'\in Y'$. 
Moreover, 
if ${\rm t}\in \{{\rm loc}, {\rm qs}, {\rm qc}\}$ and 
$p\in \cl C_{\rm t}$ then $p_{x',y'}\in \cl C_{\rm t}$. 
\end{proposition}

\begin{proof}
Let $y_1,y_2\in Y$ and $\alpha\in \frak{A}_A$. 
Since $p$ is no-signalling, we have 
\begin{eqnarray*}
p_{x',y'}(\alpha\times B|x,y_1) 
& = & 
p((\alpha\times B)\dot{\times} A'B' | xx',y_1y')\\
& = & 
p((\alpha\times A')\times BB' | xx',y_1y')\\
& = &  
p((\alpha\times A')\times BB' | xx',y_2y')
= 
p_{x',y'}(\alpha\times B|x,y_2). 
\end{eqnarray*}
It follows that $p_{x',y'}$ is a no-signalling correlation. 

Suppose that $p$ is a quantum commuting correlation, say $p = p_{E\cdot F,\xi}$, 
where $E : XX'\to (AA';H)$ and $F : YY'\to (BB';H)$ are 
operator valued channels with commuting ranges, for some Hilbert space $H$, 
and $\xi\in H$ is a unit vector; thus, 
$(E_{xx'}(\cdot))_{xx'\in XX'}$ and $(F_{yy'}(\cdot))_{yy'\in YY'}$
are Borel measurable families of quantum probability measures with commuting ranges. 
For fixed $x'\in X'$ and $y'\in Y'$, 
set $E_{x'}(\alpha) := ((E_{xx'}(\alpha\times A'))_{x\in X}$, $\alpha\in \frak{A}_A$; 
define $F_{y'}$ similarly. 
The families 
$E_{x'}$ and $F_{y'}$ of 
quantum probability measures over $A$ and $B$, respectively, 
are thus Borel measurable and have commuting ranges. 
Since $p_{x',y'} = p_{E_{x'}\cdot F_{y'},\xi}$, we have that $p_{x',y'}$ is a 
quantum commuting no-signalling correlation. 

The proof for the quantum and the local types is similar. 
\end{proof}

\begin{remark}\label{r_prodof2m}
\rm
Let $X$, $X'$, $Y$, $Y'$, $A$, $A'$, $B$ and $B'$ be compact Hausdorff spaces, and 
$\mu$ and $\nu$ (resp. $\mu'$ and $\nu'$) be Borel probability measures on $X$ and $Y$
(resp. $X'$ and $Y'$), respectively. 
If 
$p: XX'\times YY' \mapsto AA'\times BB'$ be a 
no-signalling $\mu\times \mu',\nu\times\nu'$-correlation then 
the information channel
$p_{x',y'} : X\times Y \mapsto A\times B$ 
is a no-signalling $\mu,\nu$-correlation, 
for almost all $(x',y')\in X' \times Y'$.
This follows from the fact that if a 
function $(x,x',y,y')\to h(xx',yy')$ is 
$\mu\times\mu'\times\nu\times\nu'$-measurable then 
for almost all $(x',y')$, the corresponding section 
$h_{x',y'}$ is $\mu\times\nu$-measurable
(see \cite[2.39]{folland}).

Suppose that ${\rm t}\in \{{\rm qs}, {\rm qc}\}$ 
and 
$p\in \cl C_{\rm t,sep}(\mu\times\mu',\nu\times\nu')$. 
We claim that $p_{x',y'}\in \cl C_{\rm t,sep}(\mu,\nu)$. 
The statement in the case where ${\rm t} = {\rm loc}$
is immediate. We consider the case ${\rm t} = {\rm qc}$; 
the case ${\rm t} = {\rm qs}$ is similar. 
Assume that $p = p_{E\cdot F,\xi}$, where $E$ (resp. $F$)
is an operator-valued $\mu\times\mu'$-channel
(resp. operator-valued $\nu\times\nu'$-channel) acting on $H$, 
and $\xi\in H$ is a unit vector. 
The separability of $H$ implies that for each Borel set $\alpha\subset A$ there exists a set 
$M(\alpha)\subseteq X'$ with $\mu'(M(\alpha)) = 0$, such  that 
for every $x'\not\in M(\alpha)$, the function 
$x\mapsto E_{xx'}(\alpha\times A')$ is weakly $\mu$-measurable. Choose  a countable base $\{\alpha_n\}_{n\in \mathbb N}$ for the second countable space $A$ and let $\cl R$ be the set of all finite intersections of the sets $\alpha_n$. If $M=\cup_{\alpha\in\cl R}M(\alpha)$, then $\mu(M)=0$ and the function $x\mapsto E_{xx'}(\alpha\times A')$, $\alpha\in\cl R$, is weakly $\mu$-measurable for any $x'\not\in M$. It is easy to see that the set $\cl E$ 
of all Borel subsets $\alpha\subseteq A$, such that $x\mapsto E_{xx'}(\alpha\times A')$ is weakly $\mu$-measurable  for each $x'\notin M$, is 
closed under countable increasing unions and set differences of 
the form $\alpha\setminus\alpha_0$, where $\alpha_0\subseteq \alpha$,
and hence contains the Dynkin system generated by $\cl R$. As $\cl R$ is closed under intersections, 
the Sierpi\'nski-Dynkin Theorem 
\cite[Theorem 1.13]{kn} shows that 
$\cl E$ contains the $\sigma$-algebra, generated by 
$\cl R$; thus, $\cl E = {\mathfrak A}_A$.
Similarly, there exists a set 
$N\subseteq N'$ with $\nu'(N) = 0$, such  that 
for every $y'\not\in N$, the function 
$y\to F_{yy'}(\beta\times B')$ is weakly $\nu$-measurable, 
for every Borel set $\beta\subseteq Y$. 
It follows that if $(x',y')\in M^c\times N^c$ then 
$p_{x',y'}\in \cl C_{\rm qc,sep}(\mu,\nu)$. 
\end{remark}

\begin{theorem}\label{l_prodnscor}
Let $X_i, Y_i, A_i$ and $B_i$ be compact Hausdorff spaces and
$p_i$ be a measurable no-signalling correlation over 
the quadruple $(X_i,Y_i,A_i,B_i)$, $i = 1,2$. 
\begin{itemize}
\item[(i)]
There exists a unique measurable no-signalling correlation 
$p_1\otimes p_2 = (p(\cdot|x_1x_2,y_1y_1))_{x_1x_2,y_1y_2}$ over
$(X_1X_2, Y_1Y_2, A_1A_2, B_1B_2)$, such that 
$$p(L_1\dot{\times}L_2|x_1x_2,y_1y_2) = p_1(L_1|x_1,y_1) p_2(L_2|x_2,y_2),$$
for all $L_i\in \frak{A}_{A_iB_i}$, $x_i\in X_i$, $y_i\in Y_i$, $i = 1,2$.

\item[(ii)] Let ${\rm t}\in \{{\rm loc}, {\rm qs}, {\rm qc}\}$. 
If $p_i\in \cl C_{\rm t}$, $i = 1,2$, then $p_1\otimes p_2\in \cl C_{\rm t}$. 
\end{itemize}
\end{theorem}

\begin{proof}
(i)
By Theorem \ref{cqc_prod} (applied for the case where $H = \bb{C}$), 
there exists an information channel 
$$p_1\cdot p_2 : X_1Y_1\times X_2Y_2\to A_1B_1\times A_2B_2,$$
such that 
$$(p_1\cdot p_2)(L_1\times L_2|x_1y_1,x_2y_2) = p_1(L_1|x_1,y_1) p_2(L_2|x_2,y_2),$$
for all $L_i\in \frak{A}_{A_iB_i}$, $x_i\in X_i, y_i\in Y_i$, $i = 1,2$. 
Letting $\frak{f} : A_1B_1\times A_2B_2\to A_1A_2\times B_1B_2$ be the 
flip map, given by 
$\frak{f}(a_1b_1,a_2b_2) = (a_1a_2,b_1b_2)$, define 
$$(p_1\otimes p_2)\left(L | x_1x_2,y_1y_2\right) = 
(p_1\cdot p_2)\left(\frak{f}^{-1}(L)|x_1y_1,x_2y_2\right), 
\ \ L\in \frak{A}_{A_1A_2\times B_1B_2}.$$
It is clear that $p_1\otimes p_2$ is a well-defined information channel from
$X_1X_2\times Y_1Y_2$ to $A_1A_2\times B_1B_2$; we check that it is 
a no-signalling correlation. 
Fix $x_i\in X_i, y_i,y_i'\in Y_i$, $i = 1,2$, 
$L\in \frak{A}_{A_1A_2}$. 
Write $p_i(\alpha_i|x_i)$ for the well-defined marginal value of $p_i$; 
here $\alpha_i\in \frak{A}_{A_i}$ and $x_i\in X_i$, $i = 1,2$. 
Letting $L_{a_2} = \{a_1\in A_1 : (a_1,a_2)\in L\}$, we have, by Fubini's Theorem, 
\begin{eqnarray*}
& & 
(p_1\otimes p_2)\left(L\times B_1B_2 | x_1x_2,y_1y_2\right)\\
& = & 
\int p_1(L_{a_2}\times B_1 | x_1,y_1) d p_2(a_2,b_2|x_2,y_2)\\
& = & 
\int p_1(L_{a_2} | x_1) d p_2(a_2,b_2|x_2,y_2)
= 
\int p_1(L_{a_2} | x_1) d p_2(a_2|x_2)\\
& =  & 
(p_1\otimes p_2)\left(L\times B_1B_2 | x_1x_2,y_1'y_2'\right).
\end{eqnarray*}
This shows that $p_1\otimes p_2$ is no-signalling with respect to $Y_1Y_2$; 
by symmetry, $p_1 \otimes p_2$ is no-signalling.

(ii) 
Let $H_i$ be a Hilbert space, $\xi_i\in H_i$ be a 
unit vector, and $E_i:X_i\rightarrow(A_i;H_i)$ and 
$F_i : Y_i\rightarrow(B_i;H_i)$ be operator-valued channels 
with commuting ranges, such that $p_i = p_{E_i\cdot F_i,\xi_i}$,
$i=1,2$. 
Let $E = E_1\otimes E_2$ and $F = F_1\otimes F_2$ be the 
operator-valued channels, arising from Corollary \ref{q_prod}; thus, 
$$ E : X_1X_2\mapsto (A_1A_2; H_1\otimes H_2) 
\ \mbox{ and } \ 
F : Y_1Y_2\mapsto (B_1B_2; H_1\otimes H_2).$$  
We check that $E$ and $F$ have commuting ranges;
to this end, let 
$\alpha_i\in\frak{A}_{A_i}$, $\beta_i\in\frak{A}_{B_i}$, 
$x_i\in X_i$ and $y_i\in Y_i$; then 
\begin{eqnarray*}
& & 
E(\alpha_1\times\alpha_2|x_1,x_2)F(\beta_1\times \beta_2|y_1,y_2)\\
& = & 
\Big(E_1(\alpha_1|x_1)\otimes E_2(\alpha_2|x_2)\Big)\Big(F_1(\beta_1|y_1)\otimes F_2(\beta_2|y_2)\Big)\\
& = & 
E_1(\alpha_1|x_1)F_1(\beta_1|y_1)\otimes E_2(\alpha_2|x_2)F_2(\beta_2|y_2)\\
& = & 
F_1(\beta_1|y_1)E_1(\alpha_1|x_1)\otimes F_2(\beta_2|y_2)E_2(\alpha_2|x_2)\\
& = & 
F(\beta_1\times \beta_2|y_1,y_2)E(\alpha_1\times\alpha_2|x_1,x_2).
\end{eqnarray*}
Let $\phi_{E_{x_1,x_2}} : C(A_1\times A_2)\to \cl B(H)$ and 
$\phi_{F_{y_1,y_2}} : C(B_1\times B_2)\to \cl B(H)$ be the unital completely positive 
maps, arising from the quantum probability measures $E(\cdot|x_1,x_2)$
and $F(\cdot|y_1,y_2)$ via Theorem \ref{qpm}.
Using an approximation argument as in the proof of Theorem \ref{qpm}, 
we conclude that 
$$\phi_{E_{x_1,x_2}}(f_1\otimes f_2)\phi_{F_{y_1,y_2}}(g_1\otimes g_2)=\phi_{F_{y_1,y_2}}(g_1\otimes g_2)\phi_{E_{x_1,x_2}}(f_1\otimes f_2),$$
for all $f_i\in C(A_i)$, $g_i\in C(B_i)$, $i = 1,2$.
It follows that $\phi_{E_{x_1,x_2}}$ and $\phi_{F_{y_1,y_2}}$ have commuting ranges, 
and hence, by Lemma \ref{cqpm}, 
$E(\alpha|x_1,x_2)F(\beta|y_1,y_2) = F(\beta|y_1,y_2)E(\alpha|x_1,x_2)$
for all $\alpha\in \frak{A}_{A_1}\otimes\frak{A}_{A_2}$ and all 
$\beta\in \frak{A}_{B_1}\otimes\frak{A}_{B_2}$.

Let $E\cdot F : X_1 X_2\times Y_1Y_2 \mapsto (A_1A_1\times B_1B_2; H_1\otimes H_2)$
be the operator-valued channel, arising from $E$ and $F$ via Theorem \ref{cqc_prod}. 
We have 
$$(p_1\otimes p_2)(\alpha\times \beta|x_1x_2,y_1y_2) = 
\langle E(\alpha|x_1,x_2)F(\beta|y_1,y_2)\xi_1\otimes\xi_2,\xi_1\otimes\xi_2\rangle,$$
whenever $\alpha$ and $\beta$ are Borel sets of the form 
$\alpha = \alpha_1\times\alpha_2$ and $\beta = \beta_1\times\beta_2$.
Thus, $p_1\otimes p_2 = p_{E\cdot F,\xi_1\otimes\xi_2}$, and hence 
$p_1\otimes p_2\in\cl{C}_{\rm qc}$.


The case ${\rm t} = {\rm qs}$ is similar to the case ${\rm t} = {\rm qc}$. 
Let ${\rm t} = {\rm loc}$, and assume that 
$p_1=\sum_{i=1}^{n_1}\lambda_ip_i^{(1)}\otimes p_i^{(2)}$ and $p_2=\sum_{j=1}^{n_2}\mu_jq_j^{(1)}\otimes q_j^{(2)}$ with $\lambda_i,\mu_j\in[0,1]$, $(i,j)\in\{1,\dots, n_1\}\times\{1,\dots, n_2\}$ and $\sum_{i=1}^{n_1}\lambda_i=\sum_{j=1}^{n_2}\mu_j=1$. It follows that $p_1\otimes p_2\in\cl{C}_{\rm loc}$ since
\begin{align*}
    p_1\otimes p_2 & = \Big(\sum_{i=1}^{n_1}\lambda_ip_i^{(1)}\otimes q_i^{(2)}\Big)\otimes\Big(\sum_{j=1}^{n_2}\mu_jq_j^{(1)}\otimes q_j^{(2)}\Big)\\
    & = \sum_{i=1}^{n_1}\sum_{j=1}^{n_2}\lambda_i\mu_j\Big(p_i^{(1)}\otimes p_i^{(2)}\otimes q_j^{(1)}\otimes q_j^{(2)}\Big),
\end{align*}
and $\sum_{i=1}^{n_1}\sum_{j=1}^{n_2}\lambda_i\mu_j=1$.
\end{proof}

\begin{remark}\label{r_tenpfomn}
\rm 
In the notation of Theorem \ref{l_prodnscor}, let 
$\mu_i$ (resp. $\nu_i$) be a Borel probability measure on 
$X_i$ (resp. $Y_i$), $i = 1,2$. 
It follows from Remark \ref{r_muprod} that 
if ${\rm t}\in \{{\rm loc}, {\rm qs}, {\rm qc}, {\rm ns}\}$ and 
$p_i\in \cl C_{\rm t}(\mu_i,\nu_i)$, $i = 1,2$, then $p_1\otimes p_2\in \cl C_{\rm t}(\mu_1\times\mu_2,\nu_1\times\nu_2)$.
\end{remark}


\section{Values of measurable games}\label{s_val}

In this section, we introduce measurable non-local games, 
and define two kinds of values thereof, which are 
measurable counterparts of the one-shot and the asymptotic values of a finite 
non-local game.
We show how the latter values for a finite game can be viewed as special cases 
of the former values of a measurable game, canonically arising from the given 
finite game. We start with describing the general setup, which we use
to host measurable game values.


\subsection{General setup}\label{ss_gens}

Let $X$ and $A$ be compact Hausdorff spaces.  
A \emph{measurable probabilistic hypergraph} is a pair $(\kappa,\pi)$, where 
$\kappa\in \frak{A}_{X}\otimes\frak{A}_A$ and $\pi\in P(X)$. 
A \emph{resource} over the pair $(X,A)$ is a 
convex set of information channels $p : X\to A$. 

Let $T_X : X\to X$ and $T_A : A\to A$ be homeomorphisms. 
Given $p : X\to A$, let $p^{(T_X,T_A)} : X\to A$ be the channel, given by 
$$p^{(T_X,T_A)}(\alpha | x) := p\left(T_A^{-1}(\alpha) | T_X^{-1}(x)\right), \ \ \ x\in X, \alpha\in \frak{A}_A.$$
We call a resource $\cl R$ over $(X,A)$ \emph{$(T_X,T_A)$-invariant} if 
$p^{(T_X,T_A)}\in \cl R$ whenever $p\in \cl R$.

Let 
$(\kappa,\pi)$ be a probabilistic hypergraph, where $\pi$
is $T_X$-invariant, and 
$\cl R$ be a $(T_X,T_A)$-invariant resource over the pair $(X,A)$.  
Let $p : X\mapsto A$ be an information channel.
For a susbet $M\subseteq X\times A$, let 
$M_x = \{a\in A : (x,a)\in M\}$ is the $x$-section of $M$ (note that $M_x\in \frak{A}_A$ for every $x\in X$), and 
$h_M : X\to \bb{R}$ be the function,
given by $h_M(x) = p(M_x | x)$.
Further, let 
$$\cl M = \{M\in \frak{A}_{X}\otimes\frak{A}_A : 
\mbox{ the function } h_M
\mbox{ is }\mbox{Borel}\}.
$$
It is trivial to see that, if
$\cl R$ is the set of all measurable rectangles, 
that is, the sets of the form $\chi\times\alpha$, where
$\chi\in \frak{A}_{X}$
and $\alpha\in \frak{A}_{A}$, then $\cl R\subseteq \cl M$. 
On the other hand, a direct verification shows that the set 
$\cl M$ is closed under increasing countable unions and 
set differences of the form $M\setminus N$, where $N\subseteq M$. 
By the Sierpi\'nski-Dynkin Theorem 
(see e.g. \cite[Theorem 1.13]{kn}), 
$\cl M = \frak{A}_{X}\otimes\frak{A}_A$. 
Let $\pi\otimes p$ be the probability measure on 
$(X\times A, \frak{A}_{X}\otimes\frak{A}_A)$, given by 
\begin{equation}\label{eq_piotimep}
(\pi\otimes p)(M) = \int_X p(M_x | x)d\pi(x), \ \ \ 
M\in \frak{A}_{X}\otimes\frak{A}_A.
\end{equation}
The \emph{$\cl R$-value} $\omega_{\cl R}(\kappa,\pi)$ of $(\kappa,\pi)$ is defined by letting
$$\omega_{\cl R}(\kappa,\pi) = \sup_{p\in \cl R} (\pi\otimes p)(\kappa).$$
Given $n\in \bb{N}$, let 
$\kappa^{(n)} = \cap_{k=0}^{n-1} (T_X\times T_A)^{-k}(\kappa)$,
and define \emph{inner $\cl R$-value} $\tilde{\omega}_{\cl R}(\kappa,\pi)$ 
of $(\kappa,\pi)$ by letting
\begin{equation}\label{eq_Rvalue}
\tilde{\omega}_{\cl R}(\kappa,\pi) = \limsup_{n\in \bb{N}} 
\omega_{\cl R}(\kappa^{(n)},\pi)^{1/n}.
\end{equation}

\begin{remark}\label{r_alme}
\rm 
Let $\pi$ be a Borel probability measure on $X$, and 
assume that $p : X\mapsto A$ is a $\pi$-information channel. 
Write $(p_x)_{x\in X}$ for the family of Borel measures on $A$, 
canonically associated with $p$. If $p' : X\mapsto A$ is 
another $\pi$-information channel, and if $p_x = p'_x$ for $\pi$-almost 
all $x$, then equation (\ref{eq_piotimep}) shows that 
$\pi\otimes p = \pi\otimes p'$; in other words, 
the measure $\pi\otimes p$ depends only
on the 
equivalence class of $p$ with respect to the measure $\pi$. 
\end{remark}

Remark \ref{r_alme} leads to a more 
general version of values of the probabilistic hypergraph 
$(\kappa,\pi)$ that will be useful in the sequel. 
Suppose that $p$ is a $\pi$-information channel.
Similar arguments to the ones in the paragraph before 
Remark \ref{r_alme} 
(considering the set 
$$\cl M_{\pi} = \{M\in \frak{A}_{X}\otimes\frak{A}_A : 
\mbox{ the function } h_M
\mbox{ is }\pi\mbox{-measurable}\}
$$
in the place of the set $\cl M$), 
show that, for every $M\in \frak{A}_{X}\otimes\frak{A}_A$, 
the integral (\ref{eq_piotimep}) continues to be well-defined, 
yielding a Borel measure, still denoted by $p\otimes\pi$. 

We consider a special case of particular interest. 
Let $\bb{X}$ and $\bb{A}$ be finite sets, and 
$X = \prod_{k\in \bb{Z}} \bb{X}$ and $A = \prod_{k\in \bb{Z}} \bb{A}$ 
be the Cantor spaces over $\bb{X}$ and $\bb{A}$, respectively. 
We let $X_n = \prod_{i = -n}^n \bb{X}$ and $A_n = \prod_{i = -n}^n \bb{A}$. 
Further, let $T_X : X\to X$ and $T_A : A\to A$ be the backward 
shifts.
Let $\bb{H}\subseteq \bb{X}\times\bb{A}$, considered as a hypergraph with vertex set $\bb{X}$ 
and  edges $\bb{H}_x = \{a\in \bb{A} : (x,a)\in \bb{H}\}$, indexed by $\bb{X}$. 
We consider $\bb{X}$ as embedded into $X$ as the projection onto its $0$-th coordinate.

Let 
$\kappa_{\bb{H}} = \{((x_k)_k,(a_k)_k) : (x_0,a_0)\in \bb{H}\}$; thus, 
$$\kappa_{\bb{H}}^{(n)} = \{((x_k)_k,(a_k)_k) : (x_i,a_i)\in \bb{H}, i = 0,1,\dots,n-1\}.$$
So, in this case, $\kappa_{\bb{H}}^{(n)}$ is the hypergraph 
product of $n$ copies of $\bb{H}$.
Let $\pi_0$ be a probability measure on $\bb{X}$. 
The pair $(\bb{H},\pi_0)$ gives rise to a probabilistic quantum hypergraph
in the sense of \cite[Subsection 4.3]{cltt}, by associating to $\bb{H}$
the map $\nph_{\bb{H}}$ from the projection lattice 
$\cl P_{\bb{X}}$ of the diagonal matrix algebra 
$\cl D_{\bb{X}}$ over $\bb{X}$ to the analogous projection lattice $\cl P_{\bb{A}}$, given by
$\nph_{\bb{H}}(\epsilon_{x,x}) = \sum_{a\in \bb{H}_x} \epsilon_{a,a}$
(here $(\epsilon_{x,x})_{x\in \bb{X}}$ (resp. $(\epsilon_{a,a})_{a\in \bb{A}}$)
is the set of diagonal matrix units over $\bb{X}$ (resp. $\bb{A}$)). 
One can easily see that, letting 
$\pi = \otimes_{k\in \bb{Z}} \pi_0$ be the product measure on $X$, 
the parameter $\omega_{\cl R}(\kappa_{\bb{H}},\pi)$ coincides 
with the value of the probabilistic quantum hypergraph 
$(\nph_{\bb{H}},\pi_0)$ defined in \cite{cltt}.


\subsection{Definitions and properties of game values}\label{s_defgameval}

We now specialise the concepts from Subsection \ref{s_ovic} to the case where the hypergraphs are 
non-local games. 
Let $X$, $Y$, $A$ and $B$ be compact Hausdorff spaces and $T_X$, $T_Y$, $T_A$ and $T_B$ be 
homeomorphisms of $X$, $Y$, $A$ and $B$, respectively. 
We note that the resources $\cl C_{\rm t}$ and $\cl C_{\rm t}^{\rm pr}$, 
for ${\rm t}\in \{{\rm loc}, {\rm qs}, {\rm qc}, {\rm ns}\}$ are 
$(T_X\times T_Y, T_A\times T_B)$-invariant. 
A \emph{measurable game} over $(X,Y,A,B)$ is a pair $(\kappa, \pi)$, where 
$\kappa\in \frak{A}_X\otimes \frak{A}_Y\otimes \frak{A}_A\otimes \frak{A}_B$ 
and $\pi$ is a 
$T_X\times T_Y$-invariant probability measure on $X\times Y$.
Given a measurable game $(\kappa,\pi)$ over $(X,Y,A,B)$ and 
${\rm t}\in \{{\rm loc}, {\rm qs}, {\rm qc}, {\rm ns}\}$, 
one can define several types of ${\rm t}$-values of $(\kappa,\pi)$:
the expressions 
$$\omega_{\rm t}(\kappa,\pi) = \omega_{\cl C_{\rm t}}(\kappa,\pi) 
\ \mbox{ and } \  
\tilde{\omega}_{\rm t}(\kappa,\pi) = \tilde{\omega}_{\cl C_{\rm t}}(\kappa,\pi)$$
are obtained using correlations 
of type ${\rm t}$ as the corresponding resource, while
the expressions 
$$\omega_{\rm t,sep}(\kappa,\pi) 
= \omega_{\cl C_{\rm t,sep}}(\kappa,\pi) 
\ \mbox{ and } \  
\tilde{\omega}_{\rm t,sep}(\kappa,\pi) 
= \tilde{\omega}_{\cl C_{\rm t,sep}}(\kappa,\pi)$$
are obtained using correlations 
of type ${\rm t}$ that admit realisations on 
separable Hilbert spaces.
Further, if 
$\pi = \pi_X\times\pi_Y$ is a product probability measure
(here $\pi_X$ and $\pi_Y$ are Borel probability measures on 
$X$ and $Y$, respectively), 
and ${\rm t}\in \{{\rm qs}, {\rm qc}\}$, one may consider
the values
$$\omega_{\rm t}(\kappa,\pi_X,\pi_Y) 
= \omega_{\cl C_{\rm t}(\pi_X,\pi_Y)}(\kappa,\pi) 
\ \mbox{ and } \  
\tilde{\omega}_{\rm t}(\kappa,\pi_X,\pi_Y) = 
\tilde{\omega}_{\cl C_{\rm t}(\pi_X,\pi_Y)}(\kappa,\pi),$$
and
$$\omega_{\rm t}^{\rm pr}(\kappa,\pi_X,\pi_Y) 
= \omega_{\cl C_{\rm t}^{\rm pr}(\pi_X,\pi_Y)}(\kappa,\pi) 
\ \mbox{ and } \  
\tilde{\omega}_{\rm t}^{\rm pr}(\kappa,\pi_X,\pi_Y) = 
\tilde{\omega}_{\cl C_{\rm t}^{\rm pr}(\pi_X,\pi_Y)}(\kappa,\pi).$$

We recall that a \emph{non-local game} over the quadruple 
$(\bb{X},\bb{Y},\bb{A},\bb{B})$ of finite sets is a pair
$(\bb{G},\pi_0)$, where $\pi_0$ is a 
probability measure on $\bb{X}\times \bb{Y}$ and 
$\bb{G}\subseteq \bb{X}\times \bb{Y}\times \bb{A}\times \bb{B}$, 
the latter identified with the set of 
admissible quadruples of inputs from $\bb{X}\times \bb{Y}$ and 
outputs from $\bb{A}\times \bb{B}$ in a cooperative two-player game 
(see e.g. \cite{jnvwy} or \cite{lmprsstw}).
Recall (see e.g. \cite{cltt}) that, given ${\rm t}\in \{{\rm loc}, {\rm qs}, {\rm qc}, {\rm ns}\}$, 
the \emph{${\rm t}$-value} of the game $(\bb{G},\pi_0)$ 
is the parameter 
\begin{equation}\label{eq_cprsame}
\omega_{\rm t}(\bb{G},\pi_0) = \sup_{p\in \cl C_{\rm t}} \sum_{(x,y,a,b)\in \bb{G}} \pi_0(x,y) p(a,b|x,y).
\end{equation}
Further, letting $\bb{G}^n$ be the product of $n$ copies of $\bb{G}$, that is, 
$$\bb{G}^n = \{((x_i)_{i=1}^n,(y_i)_{i=1}^n,(a_i)_{i=1}^n,(b_i)_{i=1}^n) : (x_i,y_i,a_i,b_i)\in \bb{G} \mbox{ for all } i\in [n]\},$$
the \emph{asymptotic ${\rm t}$-value} of $(\bb{G},\pi_0)$ is given by 
$$\tilde{\omega}_{\rm t}(\bb{G},\pi_0) = \lim_{n\to \infty} 
\omega_{\rm t}(\bb{G}^n,\pi_0^{\otimes n})^{1/n}$$
(supermultiplicativity and Fekete's Lemma imply that the limit on the 
right hand side exists).
One of the fundamental features of finite non-local game 
values is the value defined in (\ref{eq_cprsame}), for ${\rm t} = {\rm qs}$
or ${\rm t} = {\rm qc}$, can equivalently be defined by optimising over the sets 
$\cl C_{\rm t}^{\rm pr}$, that is, utilising PVMS's, as opposed to POVM's, 
in the representations of the corresponding resources
(this is referred to as \emph{disambiguation}). 
We next show that the same holds 
true for measurable non-local games, 
provided the corresponding probability distributions are 
non-mixed product distributions.

\begin{proposition}\label{c_disamb}
Let $X$, $Y$, $A$ and $B$ be compact Hausdorff spaces, equipped 
with homeomorphisms $T_X$, $T_Y$, $T_A$ and $T_B$, respectively, 
let $\pi_X$ (resp. $\pi_Y$) be a $T_X$-invariant (resp. 
$T_Y$-invariant) Borel probability measure on $X$ (resp. $Y$), and set $\pi = \pi_X\otimes \pi_Y$. 
Let $(\kappa, \pi_X\otimes \pi_Y)$ be a measurable game over $(X,Y,A,B)$. 
The following hold true:
\begin{itemize}
\item[(i)]
$\omega_{\rm qs}^{\rm pr}(\kappa,\pi_X,\pi_Y) = 
\omega_{\rm qs}(\kappa,\pi_X,\pi_Y)$ and
\item[(ii)]
$\tilde{\omega}_{\rm qs}^{\rm pr}(\kappa,\pi_X,\pi_Y) = \tilde{\omega}_{\rm qs}(\kappa,\pi_X,\pi_Y)$
\end{itemize}
Further, if $\pi_X$ and $\pi_Y$ are either non-atomic or totally atomic then 
\begin{itemize}
\item[(iii)]
$\omega_{\rm qc}^{\rm pr}(\kappa,\pi_X,\pi_Y) = 
\omega_{\rm qc}(\kappa,\pi_X,\pi_Y)$ and
\item[(iv)]
$\tilde{\omega}_{\rm qc}^{\rm pr}(\kappa,\pi_X,\pi_Y) = \tilde{\omega}_{\rm qc}(\kappa,\pi_X,\pi_Y)$.
\end{itemize}
\end{proposition}

\begin{proof}
(iii), (iv)
Assume that the measures $\pi_X$ and $\pi_Y$ are both non-atomic. 
Since, trivially, $\cl C_{\rm qc}^{\rm pr} \subseteq \cl C_{\rm qc}$, 
we have that 
$\tilde{\omega}_{\rm qc}^{\rm pr}(\kappa,\pi) \leq \tilde{\omega}_{\rm qc}(\kappa,\pi)$.
Let $p\in \cl C_{\rm qc}(\pi_X,\pi_Y)$, and let $H$ be a Hilbert space, 
$\xi\in H$ be a unit vector, and 
$E$ (resp. $F$) is an operator-valued $\pi_X$-information channel 
(resp. an operator-valued $\pi_Y$-information channel) acting on 
$H$, 
such that $E$ and $F$ have commuting ranges and
$p = p_{E\cdot F,\xi}$.
By Corollary \ref{c_disformu}, $E$ (resp. $F$) admits
a dilation $\tilde{E}$ (resp. $\tilde{F}$) 
that is a projection-valued 
operator-valued $\pi_X$-information channel
(resp. a projection-valued 
operator-valued $\pi_Y$-information channel), acting on $K$, 
with commuting ranges. 
Let 
$\tilde{p} = p_{\tilde{E}\cdot \tilde{F},\tilde{\xi}}$. 
We have that $\tilde{p}\in \cl C_{\rm qc}^{\rm pr}(\pi_X,\pi_Y)$. 
We have that 
\begin{equation}\label{eq_p=ptil}
\tilde{p}(\alpha\times\beta | x,y) = p(\alpha\times\beta | x,y), 
\alpha\in \frak{A}_A, \beta\in \frak{A}_B, (x,y)\in X\times Y.
\end{equation}
Since the sets of the form $\alpha\times\beta$, for 
$\alpha\in \frak{A}_A$ and $\beta\in \frak{A}_B$, generate the $\sigma$-algebra
$\frak{A}_A\otimes\frak{A}_B$, equation (\ref{eq_p=ptil}) implies that 
$$\tilde{p}(L | x,y) = p(L | x,y), \ \ 
L\in \frak{A}_A\otimes \frak{A}_B, (x,y)\in X\times Y,$$
that is, $p = \tilde{p}$. 
Thus, 
\begin{eqnarray*}
\omega_{\rm t}(\kappa,\pi_X,\pi_Y) 
&\geq& 
(\pi\otimes \tilde{p})(\kappa)
= 
\int_X\int_Y \tilde{p}(\kappa_{(x,y)} | x,y) d\pi_Y(y) d\pi_X(x)\\
&=& \int_X\int_Y p(\kappa_{(x,y)} | x,y) d\pi_Y(y) d\pi_X(x).
\end{eqnarray*}
Taking the supremum over all $p$, we obtain the inequality 
$\tilde{\omega}_{\rm qc}^{\rm pr}(\kappa,\pi_X,\pi_Y) \geq \tilde{\omega}_{\rm qc}(\kappa,\pi_X,\pi_Y)$. The totally atomic case is treated similarly.

(i), (ii)
Let $p\in \cl C_{\rm qs}(\pi_X,\pi_Y)$, and write $p = p_{E\otimes F,\xi}$, where $E$ (resp. $F$)
is an operator-valued $\pi_X$-information 
channel 
(resp. an operator-valued $\pi_Y$-information 
channel), acting on the Hilbert space $H$ (resp. $K$), 
and $\xi\in H\otimes K$ is a unit vector. 
Using Theorem \ref{l_measfree_comm}, let $\tilde{E}$ (resp. $\tilde{F}$) be a projection-valued
$\pi_X$-information channel
(resp. a projection-valued $\pi_Y$-information channel), acting on a Hilbert space $\tilde{H}$ (resp. $\tilde{K}$),
and $V : H\to \tilde{H}$ (resp. $W : K\to \tilde{K}$) be an isometry, such that 
$E(\alpha|x) = V^* \tilde{E}(\alpha|x)V$ 
(resp. $F(\beta|y) = W^* \tilde{F}(\beta|y)W$)
for $\pi_X$-almost all $x$ (resp. $\pi_Y$-almost all $y$). 
Letting $\tilde{\xi} = (V\otimes W)\xi$, we now have that 
the NS correlation $p_{\tilde{E}\otimes \tilde{F},\tilde{\xi}}$
is in the same equivalence class with $p$ with respect to the measure 
$\pi_X\times\pi_Y$; the statement now follows from Remark \ref{r_alme}.

\end{proof}



Given measure spaces $X_i$, $Y_i$ and Borel probability measures
$\pi_i$ on $X_i\times Y_i$, $i = 1,2$, we let 
$\pi_1\dot{\times}\pi_2$ be the Borel probability measure on 
$X_1X_2\times Y_1Y_2$, with the property 
$$(\pi_1\dot{\times}\pi_2)(\alpha_1\times\alpha_2\times\beta_1\times\beta_2)
= \pi_1(\alpha_1\times\beta_1) \pi_1(\alpha_2\times\beta_2),$$
for all $\alpha_i\in \frak{A}_{X_i}$ and all $\beta_i\in \frak{A}_{Y_i}$, $i = 1,2$.

\begin{proposition}\label{p_prodva}
Let $X_i, Y_i, A_i$ and $B_i$ be compact Hausdorff spaces, 
$(\kappa_i,\pi_i)$ be a measurable game over $(X_i,Y_i,A_i,B_i)$, $i = 1,2$, and 
${\rm t}\in \{{\rm loc}, {\rm qs}, {\rm qc}, {\rm ns}\}$. 
Then 
\begin{equation}\label{eq_eqprodco0}
\omega_{\rm t}(\kappa_1\dot{\times} \kappa_2,\pi_1\dot{\times}\pi_2) 
\geq \omega_{\rm t}(\kappa_1,\pi_1)
\omega_{\rm t}(\kappa_2,\pi_2).
\end{equation}
Furthermore, if $\kappa_2 = X_2Y_2\times A_2B_2$ then 
\begin{equation}\label{eq_eqprodco}
\omega_{\rm t}(\kappa_1\dot{\times} \kappa_2,\pi_1\dot{\times}\pi_2) 
= \omega_{\rm t}(\kappa_1,\pi_1).
\end{equation}
\end{proposition}

\begin{proof}
By Theorem \ref{l_prodnscor}, for any $p_i\in\cl C_{\rm t}$, $i=1,2$,
$$\omega_{\rm t}(\kappa_1\dot{\times} \kappa_2,\pi_1\dot{\times}\pi_2)\geq ((\pi_1\dot{\times}\pi_2)\otimes (p_1\otimes p_2))(\kappa_1\dot{\times} \kappa_2)
= (\pi_1\otimes p_1)(\kappa_1)(\pi_2\otimes p_2)(\kappa_2),$$
giving the inequality.

To see the equality, using Proposition \ref{l_reduction}, we observe that 
\begin{eqnarray*}
&&
\hspace{-0.3cm}
\omega_{\rm t}(\kappa_1\dot{\times} \kappa_2, \pi_1\dot{\times}\pi_2)
=
\sup_{p\in\cl C_{\rm t}}(\pi_1\dot{\times}\pi_2\otimes p)(\kappa_1\dot{\times}\kappa_2)\\
&&
\hspace{-0.3cm} =\sup_{p\in\cl C_{\rm t}}\int_{X_1X_2\times Y_1Y_2}p((\kappa_1\dot{\times} \kappa_2)_{(x_1x_2, y_1y_2)}|x_1x_2,y_1y_2)d(\pi_1\dot{\times}\pi_2)(x_1x_2,y_1y_2)\\
&&
\hspace{-0.3cm} =\sup_{p\in\cl C_{\rm t}}\int_{X_1X_2\times Y_1Y_2}p((\kappa_1)_{x_1,y_1}\times A_2B_2|x_1x_2,y_1y_2)d(\pi_1\dot{\times}\pi_2)(x_1x_2,y_1y_2)\\
&&
\hspace{-0.3cm}=
\sup_{p\in\cl C_{\rm t}}\int_{X_2\times Y_2}\int_{X_1\times Y_1}p_{x_2,y_2}((\kappa_1)_{x_1,y_1}|x_1,y_1)d\pi_1(x_1,y_1))d\pi_2(x_2,y_2)\\
&&\hspace{-0.3cm} 
\leq\omega_{\rm t}(\kappa_1,\pi_1)\int_{X_2\times Y_2}d\pi_2(x_2,y_2)
= \omega_{\rm t}(\kappa_1,\pi_1).
\end{eqnarray*}
The reverse inequality follows from (\ref{eq_eqprodco0}). 
\end{proof}

\begin{remark}\label{r_sepao}
\rm 
Using Remark \ref{r_prodof2m}, we note that,
in the notation of Proposition \ref{p_prodva}, we have that 
\begin{equation}\label{eq_eqprodco01}
\omega_{\rm t,sep}(\kappa_1\dot{\times} \kappa_2,\pi_1\dot{\times}\pi_2) 
\geq \omega_{\rm t,sep}(\kappa_1,\pi_1)
\omega_{\rm t,sep}(\kappa_2,\pi_2)
\end{equation}
and, if $\kappa_2 = X_2Y_2\times A_2B_2$, then 
\begin{equation}\label{eq_eqprodco23}
\omega_{\rm t,sep}(\kappa_1\dot{\times} \kappa_2,\pi_1\dot{\times}\pi_2) 
= \omega_{\rm t,sep}(\kappa_1,\pi_1).
\end{equation}
\end{remark}

We next show that the value and the inner value of a 
measurable game include as special cases 
the value and the asymptotic value of a non-local game over a quadruple of 
finite sets. 
Suppose that $\bb{X}$, $\bb{Y}$, $\bb{A}$ and $\bb{B}$ are finite sets,
and $(\bb{G},\pi_0)$ is a non-local game over the quadruple 
$(\bb{X},\bb{Y},\bb{A},\bb{B})$.
Set $X = \prod_{k\in \bb{Z}} \bb{X}$, $Y = \prod_{k\in \bb{Z}} \bb{Y}$, 
$A = \prod_{k\in \bb{Z}} \bb{A}$ and $B = \prod_{k\in \bb{Z}} \bb{B}$. 
We associate with the set $\bb{G}$ the subset $\kappa_{\bb{G}}\subseteq X\times Y \times A\times B$
as in Subsection \ref{s_ovic}, and we let $\pi = \otimes_{k\in \bb{Z}} \pi_0$; 
clearly, $\pi$ is a shift-invariant Borel probability measure on $X\times Y$.

\begin{theorem}\label{p_coincides}
For 
${\rm t}\in \{{\rm loc},{\rm qs}, {\rm qc},{\rm ns}\}$, we have that
$\omega_{\rm t}(\kappa_{\bb{G}},\pi) = \omega_{\rm t}(\bb{G},\pi_0)$ and 
$\tilde{\omega}_{\rm t}(\kappa_{\bb{G}},\pi) = \tilde{\omega}_{\rm t}(\bb{G},\pi_0)$. 
Moreover, if
${\rm t}\in \{{\rm qs}, {\rm qc}\}$ then 
$\omega_{\rm t,sep}(\kappa_{\bb{G}},\pi) = \omega_{\rm t}(\bb{G},\pi_0)$ and 
$\tilde{\omega}_{\rm t,sep}(\kappa_{\bb{G}},\pi) = \tilde{\omega}_{\rm t}(\bb{G},\pi_0)$.
\end{theorem}

\begin{proof}
Let  $X_n=\prod_{k\in\mathbb Z\setminus\{0,1,\ldots,n\}}\mathbb X$ and similar define $Y_n$, $A_n$ and $B_n$.
Then $\kappa_{\bb G} = \bb G\dot{\times} (X_0Y_0\times A_0B_0)$. By Proposition \ref{p_prodva}, $\omega_{\rm t}(\kappa_{\bb G},\pi)=\omega_{\rm t}(\bb G,\pi_0)$ in the case 
${\rm t} \in \{{\rm loc},{\rm ns}\}$ and, by Remark \ref{r_sepao}, 
$\omega_{\rm t,sep}(\kappa_{\bb G},\pi)=\omega_{\rm t}(\bb G,\pi_0)$ in the case 
${\rm t} \in \{{\rm qs},{\rm qc}\}$.
Similarly, $\kappa_{\bb G}^{(n)} = {\bb G}^{(n)}\dot{\times}(X_nY_n\times A_nB_n)$, giving $\omega_{\rm t}(\kappa_{\bb G}^{(n)},\pi)=\omega_{\rm t}({\bb G}^{(n)},\pi_0^{\otimes n})$
for ${\rm t} \in \{{\rm loc},{\rm ns}\}$, and 
$\omega_{\rm t,sep}(\kappa_{\bb G}^{(n)},\pi)=\omega_{\rm t}({\bb G}^{(n)},\pi_0^{\otimes n})$
for ${\rm t} \in \{{\rm qs},{\rm qc}\}$, 
and the claim is proved.
\end{proof}


\subsection{An example: a game of Markov type}\label{ss_examples}

In this subsection, we give a class of examples of 
non-local games with a one-site memory, 
whose asymptotic values are strictly smaller than their inner values. 
Let 
$\bb{X}, \bb{Y}, \bb{A}$ and $\bb{B}$ be finite sets 
and 
$\bb{G} \subseteq \bb{X}\times \bb{Y}\times \bb{A}\times \bb{B}$. 
Set 
$$\hspace{-4cm}\frak{m}_{\bb{G}} = \{(x_i)_{i\in \bb{Z}},
(y_i)_{i\in \bb{Z}},(a_i)_{i\in \bb{Z}},(b_i)_{i\in \bb{Z}}) : $$
$$ \hspace{4cm} (x_0,y_0,a_0,b_0)\in \bb{G} \mbox{ or } (x_1,y_1,a_1,b_1)\in \bb{G}\}.$$

\begin{theorem}\label{th_markov}\label{th_markov}
Let 
$\bb{X}, \bb{Y}, \bb{A}$ and $\bb{B}$ be finite sets, 
$\pi_0 : \bb{X}\times \bb{Y} \to [0,1]$ be a probability distribution, 
$\bb{G} \subseteq \bb{X}\times \bb{Y}\times \bb{A}\times \bb{B}$, and 
${\rm t}\in \{{\rm loc}, {\rm qs}, {\rm qc}, {\rm ns}\}$.
Set $\pi = \otimes_{i\in \bb{Z}} \pi_0$. 
Assume that $\tilde{\omega}_{\rm t}(\bb{G},\pi_0) < 1$. 
Then 
$\tilde{\omega}_{\rm t}(\frak{m}_{\bb{G}},\pi) 
\geq 
\sqrt{\tilde{\omega}_{\rm t}(\bb{G},\pi_0)}$. 
\end{theorem}

\begin{proof}
Write $\bb{F} = \bb{X}\times \bb{Y} \times \bb{A} \times \bb{B}$. 
We have that 
$$\frak{m}_{\bb{G}}^{(n)} = \{(x_i)_{i\in \bb{Z}},
(y_i)_{i\in \bb{Z}},(a_i)_{i\in \bb{Z}},(b_i)_{i\in \bb{Z}}) : 
\mbox{ for every } i = 0,\dots,n-1, $$
$$
\hspace{3cm} \mbox{ either } (x_i,y_i,a_i,b_i)\in \bb{G} \mbox{ or } 
(x_{i+1},y_{i+1},a_{i+1},b_{i+1})\in \bb{G}\}.$$
Thus, for an even $n\in \bb{N}$, say $n = 2k$, we have 
$$\kappa_{\bb{G}\dot{\times}\bb{F}}^{(k)} \subseteq \frak{m}_{\bb{G}}^{(n)}.$$
By Proposition \ref{p_prodva} and Theorem \ref{p_coincides}, 
$$
\omega_{\rm t}(\bb{G}^k,\pi_0^{\otimes k})
\leq 
\omega_{\rm t}(\frak{m}_{\bb{G}}^{(n)},\pi);$$
thus, 
$$
\sqrt{\omega_{\rm t}(\bb{G}^k,\pi_0^{\otimes k})^{\frac{1}{k}}} = 
\omega_{\rm t}(\bb{G}^k,\pi_0^{\otimes k})^{\frac{1}{n}}
\leq 
\omega_{\rm t}(\frak{m}_{\bb{G}}^{(n)},\pi)^{\frac{1}{n}}.$$
It follows that 
\begin{eqnarray*}
\sqrt{\tilde{\omega}_{\rm t}(\bb{G},\pi_0)} 
& = &  
\lim_{k\to\infty} \sqrt{\omega_{\rm t}(\bb{G}^k,\pi_0^{\otimes k})^{\frac{1}{k}}} 
\leq 
\limsup_{k\in \bb{N}}\omega_{\rm t}(\frak{m}_{\bb{G}}^{(n)},\pi)^{\frac{1}{n}}\\
& \leq & 
\limsup_{m\in \bb{N}}\omega_{\rm t}(\frak{m}_{\bb{G}}^{(m)},\pi)^{\frac{1}{m}}
= 
\tilde{\omega}_{\rm t}(\frak{m}_{\bb{G}},\pi).
\end{eqnarray*}
\end{proof}


\section{Questions}\label{s_questions}

The disambiguation result Proposition \ref{c_disamb}, in the case of quantum commuting values, 
is currently only available for continuous or totally atomic measures. In addition, 
the passage to the projective realisation of the correlations, utilising spectral measures as 
opposed to arbitrary quantum probability measures, requires a non-separable Hilbert space. 
For the resolution of these issues, one would need to establish 
an analogue of Theorem \ref{disambiguation} hold for not necessarily continuous measures $\mu$ and $\nu$; 
such an analogue can also be viewed as a bivariate version of Theorem \ref{l_measfree_comm}.

\begin{question}
\rm 
Let $X$ and $Y$ be compact Hausdorff spaces, and $\mu$ and $\nu$ be Borel probability measures on $X$ and $Y$,
respectively. 
Further, let $\cl A$, $\cl B$ separable $C^*$-algebras, and $\phi:\cl A\to L^\infty(X,\mu,\cl B(H))$  and $\psi:\cl B\to L^\infty(Y,\nu,\cl B(H))$ be unital completely positive maps with commuting ranges. 
Do there exist a 
  Hilbert space $K$, an isometry $V : H\to K$, a $\mu$-measurable family $(\pi_x)_{x\in X}$  and a $\nu$-measurable family $(\rho_y)_{y\in Y}$ of 
		unital *-representations of $\cl A$  and $\cl B$ on $K$ respectively with commuting ranges, such that if $u\in \cl A$ and $v\in\cl B$, then
		$$\phi(u)(x)\psi(v)(y) = V^*\pi_x(u)\rho_y(v) V, \ \ \text{ $\mu\times\nu $-almost everywhere?}$$
\end{question}

\smallskip

In the case of finite non-local games without memory, an application of Fekete's Lemma shows that the asymptotic 
value of the game can be obtained as a limit of an increasing sequence 
of normalised values of $n$-fold product games. 
Currently we do not know of this is the case in the 
general case we propose in Section \ref{s_val}:

\begin{question}\label{q_limit}
\rm 
Under what conditions can we replace the limsup in the definition of the inner 
$\cl R$-value of probabilistic measurable 
hypergraph in (\ref{eq_Rvalue}) with a limit?
It suffices to have supermultiplicativity in the sense that 
$$\omega_{\cl R}(\kappa^{(n+m)},\pi) \geq \omega_{\cl R}(\kappa^{(n)},\pi)\omega_{\cl R}(\kappa^{(m)},\pi).$$
We note that this would involve properties both of the resource $\cl R$
and the measure $\pi$. 
\end{question}

\smallskip

It has been of substantial interest, in both the classical \cite{raz} and 
the quantum \cite{csuu} case to study the growth behaviour of the values of 
iterated copies of a finite non-local game. In some cases, a perfect parallel 
repetition result holds for the quantum value. 

\begin{question}
\rm 
Is there a parallel repetition theorem for any of the value types of a 
non-local game with memory? 
\end{question}

\smallskip

\medskip

\subsection*{Acknowledgements}
The authors would like to thank 
Andreas Winter for useful discussions on games with memory. 
The first two authors were supported by 
NSF grants CCF-2115071 and DMS-2154459. The third author was supported by the Swedish Research Council. 
The authors thank the Department of Mathematical Sciences 
of the University of Delaware for funding allowing the visit of the third author
to Delaware in Spring 2024.


\end{document}